\documentclass[aps,pra,twocolumn,10pt,superscriptaddress]{revtex4-1}

\usepackage[english]{babel}
\usepackage[T1]{fontenc} %for what? That's why: http://tex.stackexchange.com/questions/664/why-should-i-use-usepackaget1fontenc
\usepackage{times}

\usepackage{amsthm} %for theorem environments
\usepackage{amssymb} %lots of math symbols
\usepackage{amsmath} %for spacing in equations and aligns
\usepackage{mathbbol} %for \1, load at last
\usepackage[normalem]{ulem}
\usepackage[usenames,dvipsnames]{color}
\usepackage{graphicx}
\usepackage{overpic}
\usepackage{paralist} %for compactenum
\usepackage{mathtools} %for, e.g., \DeclareMathOperator
\usepackage{complexity}

\newtheorem{theorem}{Theorem}
\newtheorem{lemma}[theorem]{Lemma}

\newtheorem{definition}[theorem]{Definition}
\newcommand{\e}{\mathrm{e}}
\renewcommand{\i}{\mathrm{i}}
\newcommand{\1}{\mathbb{1}}
\DeclareMathOperator{\landauO}{\mathrm{O}}
\DeclareMathOperator{\Tr}{Tr}
\newcommand\intd{\mathrm{d}}
\DeclareMathOperator\fid{F}
\newcommand{\id}{\mathrm{id}}
\DeclareMathOperator{\rank}{\operatorname{rank}}
\renewcommand{\Re}{\operatorname{Re}}
\newcommand{\imag}{\Im \text{m}}
\newcommand{\Trot}{\mathrm{Trot}}

\newcommand{\mc}[1]{\mathcal{#1}}
\renewcommand{\L}{\mc{L}}
\renewcommand{\K}{\mc{K}}

\newcommand{\ket}[1]{\left|{#1}\right\rangle}
\newcommand{\Lket}[1]{\left|{#1}\right\rangle\!\rangle}

\newcommand{\bra}[1]{\left\langle{#1}\right|}

\newcommand{\braket}[2]{\left\langle #1 \middle| #2 \right\rangle}

\newcommand{\ketbra}[2]{\ket{#1} \!\! \bra{#2}}

\newcommand{\Scp}[2]{\ensuremath{\, \langle #1 \,, #2 \,\rangle}}
\newcommand\Abs[1]{\ensuremath{| #1|}}
\newcommand{\norm}[1]{\left\Vert #1\right\Vert}
\newcommand{\normn}[1]{\lVert #1\rVert}
\newcommand{\normb}[1]{\bigl\Vert #1\bigr\Vert}
\newcommand{\dnorm}[1]{\left\Vert #1\right\Vert_{\diamond}}
\newcommand{\dnormb}[1]{\bigl\Vert #1\bigr\Vert_{\diamond}}

\newcommand{\argdot}{{\; \cdot \;}}
\newcommand{\kw}[1]{\frac{1}{#1}}
\newcommand{\tkw}[1]{\tfrac{1}{#1}}
\newcommand{\Heven}{H_{e}}
\newcommand{\Hodd}{H_{o}}
\newcommand{\Heo}{H_{\nu}}
\newcommand{\Aeven}{\mathcal{H}_{e}}
\newcommand{\Aodd}{\mathcal{H}_{o}}
\newcommand{\Aeo}{\mathcal{H}_{\nu}}
\newcommand{\Deven}{\mathcal{D}_{e}}
\newcommand{\Dodd}{\mathcal{D}_{o}}
\newcommand{\Infid}{\mathcal{I}}
\newcommand{\locT}[1][l]{A^{[#1]}}
\newcommand{\locTel}[2]{A^{#1}_{#2}}
\newcommand{\dw}{\delta}
\newcommand{\app}[1]{App.~\ref{#1}}

\definecolor{martin}{rgb}{0,.1,1}

\definecolor{jens}{RGB}{40,180,40}

\definecolor{ulmcolor}{rgb}{.4,0,.6}

\definecolor{ulmmethod}{rgb}{0,.3,.6}

\definecolor{scalingcol}{rgb}{.7,.3,0}

\definecolor{albert}{rgb}{0.77,0.29,0.55}

\definecolor{dangercolor}{rgb}{0.8,0.,0.}

\usepackage[normalem]{ulem} % for \sout
\usepackage{cancel}

\newcommand{\fu}{Dahlem Center for Complex Quantum Systems, Freie Universit\"{a}t Berlin, 14195 Berlin, Germany}
\newcommand{\ulm}{Institute for Complex Quantum systems (ICQ), Universit{\"a}t Ulm, 89069 Ulm, Germany}
\newcommand{\colorado}{Department of Physics, Colorado School of Mines, Golden, Colorado 80401, USA}

\begin{document}

\title{A positive tensor network approach for simulating open quantum many-body systems}

\author{A.\ H.\ Werner}
\affiliation{\fu}
\author{D.\ Jaschke}
\affiliation{\ulm}
\affiliation{\colorado}
\author{P.\ Silvi}
\affiliation{\ulm}
\author{M.\ Kliesch}
\affiliation{\fu}
\author{T.\ Calarco}
\affiliation{\ulm}
\author{J.\ Eisert}
\affiliation{\fu}
\author{S.\ Montangero}
\affiliation{\ulm}

\begin{abstract}
Open many-body quantum systems play an important role in quantum optics and condensed-matter physics,
and capture phenomena like transport,
interplay between Hamiltonian and incoherent
dynamics, and topological order generated by dissipation.
We introduce a versatile and practical method to numerically simulate one-dimensional
open quantum many-body dynamics using tensor networks. It is based on representing mixed quantum states in a locally purified form, which
 guarantees that positivity is preserved at all times. Moreover, the approximation error is controlled with respect to the trace norm.
Hence, this scheme overcomes various obstacles of the known numerical open-system evolution schemes.
To exemplify the functioning of the approach, we study both stationary states and transient dissipative behaviour,
for various open quantum systems ranging from few to many bodies.
\end{abstract}

\maketitle

Open quantum systems are ubiquitous in physics.
To some extent any quantum system is coupled to an environment,
and in many instances this interaction significantly alters the system's dynamics.
Traditionally such decoherence processes are seen as an enemy to coherent state manipulation.
However, suitably engineered dissipation can also have beneficial effects and can be exploited for state preparation
\cite{DissipationZoller,DissipationComputing,DissipationEntanglement,DissipationOld,DissipationCritical,Timing},
even of states containing strong entanglement or featuring topological order \cite{DissipationTopologyZoller,TopologyDissipationLong}.
In condensed matter physics, many concepts such as transport are often studied within the closed systems paradigm, but, it is becoming
increasingly  clear that some familiar concepts may have to be revisited in the open system setting  \cite{ProsenDissipativeStationaryXXZ}, where the interplay between coherent quantum many-body and open systems dynamics, i.e. the competition between Hamiltonian interactions and dissipation leads to interesting physical effects.
Since few analytical methods are available for such systems, the design of novel numerical tools for the simulation of dissipative quantum many-body systems is of the utmost importance.
In this work, we present a new algorithm which captures the open many-body dynamics in one spatial dimension
-- for both transient and steady regimes -- based on a locally purified tensor network ansatz-class. It comprises a new approach in that the positivity of the operators is maintained  during the whole simulation.
Importantly, the approximation errors can be controlled in a way that yields a trace-norm certificate.
Hence, the algorithm provides not only a conceptually
new approach to the problem, but also combines several desired features of the existing schemes and overcomes previous limitations.

Tensor-network ansatz-classes
have proven to be widely successful to capture the physics of many-body states
\cite{FCS,MPSReps,Murg,MPSRev,Flow,AreaReview}. They rely on the idea that relevant quantum states lie in the
very small sub-manifold with local correlations, which in turn can be efficiently captured in terms of tensor networks.
The density-matrix renormalisation method \cite{DMRGWhite92} can indeed be viewed as a variational principle over {\it matrix-product states} \cite{MPSZero,MPSReps,MPSRev,Lagrangian,orus2013}.
Generalising these ideas, a number of exciting methods have been proposed \cite{Renegade,RizMera,Adaptive,MERAF1,MERAF3,MERArealaw,Lubasch2014}, some of which
also allow to study open quantum systems. In most cases
{\it matrix-product operators} (MPO) are at the heart of these methods. Indeed, several variants have already been developed
\cite{Mixed,ZwolakVidal,LaeuchliDissipative,DaleyOpenReview,KarraschDrude,KarraschTransport,PizornfiniteTemp,JianMariCarmen,CaiBarthel},
many of which exploit
the well-known features of tensor network ansaetze to encode the mixed many-body quantum states in a compact
matrix-product formulation, ultimately making the algorithm efficient and stable
both for transient \cite{Mixed,ZwolakVidal,LaeuchliDissipative} and steday state physics \cite{JianMariCarmen}.

However, in such an MPO description, the resulting truncated operators may not be positive;
in fact this property can not even be locally tested, because it turns out to be a computationally intractable problem \cite{UndecidableMPO}.
In Ref.~\cite{Mixed} this problem is circumvented by dropping the positivity assumption during the time-evolution,
which requires that the
approximation errors remain sufficiently small.
Alternatively, quantum jumps schemes make use of a stochastic unravelling of the master equation
\cite{DaleyZoller,DissipativeLightScattering,DaleyOpenReview} and
then employ pure-state techniques, at the expense of having to sample over many realisations.
Comparative performance of these two approaches has been recently investigated \cite{LauchCompare}.

Remarkably, the subset of matrix-product operators that are cast in a locally purified form \cite{Mixed,Gemma} shows promising features:
Such operators are positive by construction and
exhibit all the helpful features typical to tensor networks.
However, while variational
algorithms optimising within this class have been developed for
two-dimensional projected entangled pair states \cite{Pizorn2011time,Lubasch2014,Lubasch2014b}
a practical algorithm for
one-dimensional open systems has yet to be formulated.
Here we show that such a positivity-preserving algorithm can actually be engineered for Markovian dynamics:
This scheme has the computational efficiency of tensor network methods, allows to control all approximation errors in the operationally relevant trace norm, and  preserves positivity
by construction, thus ultimately merging the advantages of previous techniques while
solving known issues.
\begin{figure*}
\includegraphics[height=5.5cm,angle=270,clip=true, trim = 1cm 0cm 1cm 0.04cm]{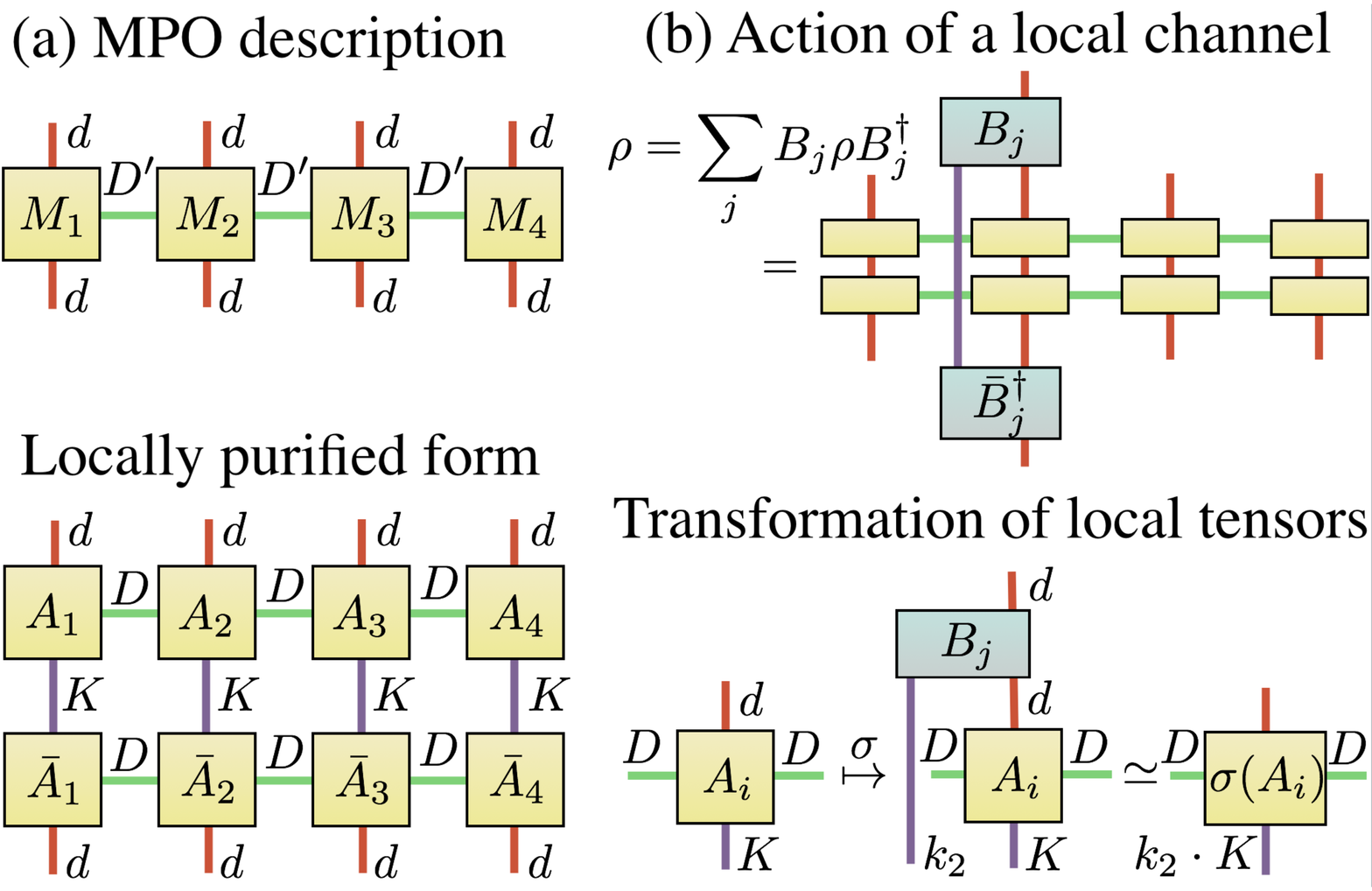}
$\quad$
\includegraphics[height=5.5cm,angle=270,clip=true, trim = 1cm 0cm 1cm 0.01cm]{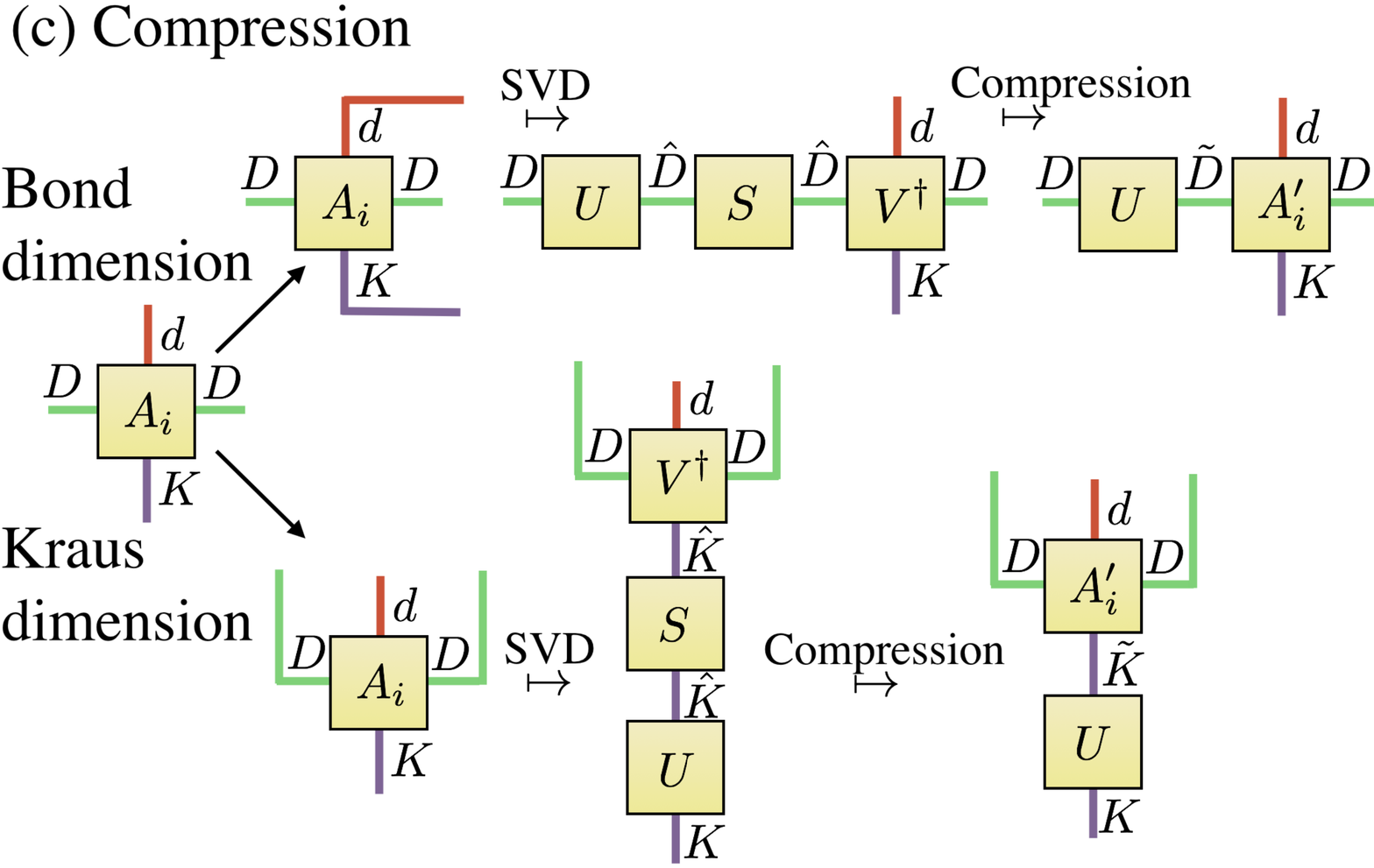}
$\quad$
\includegraphics[height=5.5cm,angle=270,clip=true, trim = 1cm 0cm 1cm 0cm]{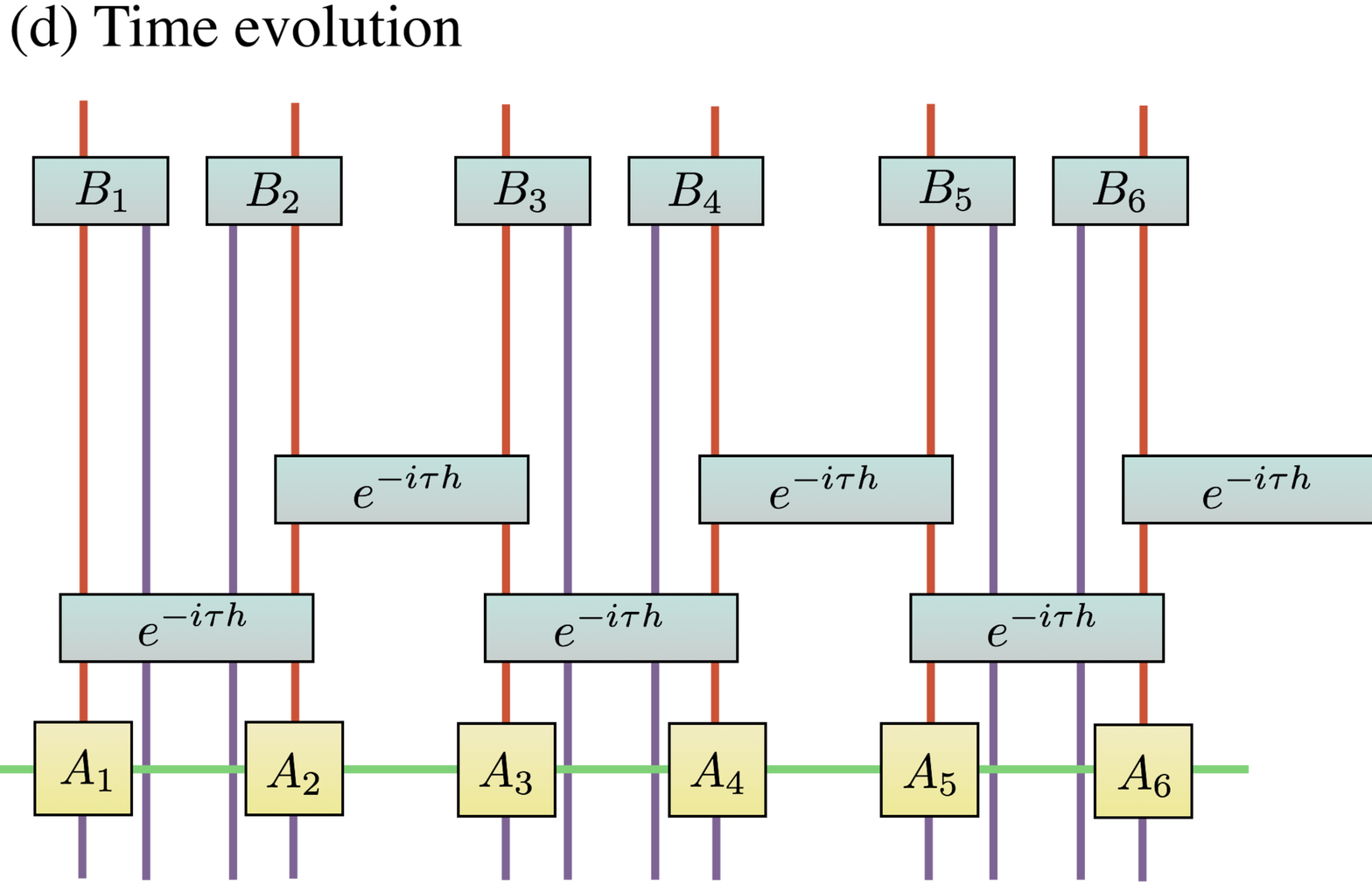}
  \caption{{Markov dynamics of a quantum spin chain on the level of local tensors. \emph{a)} shows the relationship between a density matrix $\rho$
  in MPO representation (top) and the locally purified tensor network (bottom) with tensors $A_l$, physical dimension $d$, bond dimension $D'$ and Kraus dimension $K$.
\emph{b)} Action of a local channel $T$ acting exclusively on lattice site $2$ on the level of the MPO and on the level of the locally purified tensors.
In the latter, the Kraus rank $k_2$ of the quantum channel $T$ is joined together with $K$.
\emph{c)} Compression schemes for the bond and Kraus dimension of a local tensor via singular value decompositions (SVD).
$d)$ Locally purified evolution of a time step $\e^{\tau\mathcal L}$ for a nearest neighbor Hamiltonian and on-site Lindblad operators.
We only show three of the five Trotter-Suzuki layers from Eq.~\eqref{eq:secorder}.}}\label{fig:sketchy}
\end{figure*}
%
%
%%%%%%%%%%%%%%%%%%%%%%%%%%%%%%%%%%%%%%%%%%%%%%%%%%%%%%%%%%%%%%%%%%%%%%%%%%%%%%%%
% METHODS
%%%%%%%%%%%%%%%%%%%%%%%%%%%%%%%%%%%%%%%%%%%%%%%%%%%%%%%%%%%%%%%%%%%%%%%%%%%%%%%%
%

{\it Algorithm.} Our goal is to simulate the transient evolution as well as to find steady states of spin chains
under local Markovian dynamics, i.e.~one-dimensional lattice systems (at finite system size, with open boundary conditions) governed by the Lindblad master equation
\begin{equation}\label{eq:Lbld}
  \frac{\intd \rho}{\intd t} = \mathcal L(\rho) = -\i[H,\rho] + \mathcal{D}\left(\rho\right)\, .
\end{equation}
Here, $H=\sum_{j} H_j$ is the local Hamiltonian of the system and the dissipative part of the Lindblad generator $\mathcal{L}$ takes the form
\begin{equation}
\mathcal{D}\left(\rho\right)=  \sum_{\alpha} \left(L_\alpha \rho L_\alpha^\dagger -\frac{1}{2}\{L_\alpha^\dagger L_\alpha, \rho \} \right)\;,
\end{equation}
where the Lindblad operators $L_\alpha$ model the coupling of the system with the environment. We focus on the typical scenario, where
the  elementary Hamiltonian terms $H_j$  as well as the Lindblad  operators $L_\alpha$ are two-local, meaning that they only couple spins on neighboring sites and denote them by $H^{[l,l+1]}$  or $L^{[l,l+1]}$, respectively.

We describe the variational mixed state of the system as a tensor network representing the density matrix $\rho$.
But instead of expressing $\rho$ directly as an MPO \cite{MPO_Representations, Mixed}
we keep it at every stage of our
algorithm in its locally purified form
$\rho = X X^{\dagger}$, where the \emph{purification operator} $X$ is decomposed as a variational tensor network
\begin{equation}\label{eq:X_in_terms_of_A}
[X]^{s_1 ,\ldots, s_N}_{r_1 ,\ldots, r_N} = \!\!\!\!\sum_{m_1 {,} \ldots {,} m_{N-1}}\!\!\!\! \locTel{[1]s_1,r_1}{m_1}
\locTel{[2]s_2,r_2}{m_1,m_2}\ldots\locTel{[N]s_N,r_N}{m_{N-1}} ,
\end{equation}
with $1\leq s_l \leq d, 1 \leq  r_l \leq K$, and $1 \leq m_l \leq D$.
That is, we represent $\rho$ by a \emph{locally purified tensor network}
made of rank-four tensors $\locT$ with physical dimension $d$, bond dimension $D$
and Kraus dimension $K$ (shown in Fig.~\ref{fig:sketchy}a). Our algorithm is an extension of the Time Evolving Block Decimation (TEBD) scheme \cite{GVidal},
acting on the level of the local tensor $\locT$ that also allows for dissipative channels,
and never requires to contract the two tensor network layers ($X$ and $X^{\dagger}$) together.
Similarly to TEBD,
it involves splitting the propagator $\e^{\tau\mathcal L}$ for a small time-step $\tau$ into several dissipative Trotter-Suzuki layers \cite{KliBarGog11} of mutually commuting operators. Let us consider the evolution from time $t$ to $t+\tau$ in row-wise vectorization
\begin{equation}
\Lket{\rho_{t+\tau}} = \Lket{\e^{\tau \mathcal L} \rho_t} =  \e^{\tau\left(-\i H \otimes \1 + \i\1 \otimes \bar{H} + \mathcal{D} \right)}\Lket{\rho_t}\; ,
\end{equation}
where $\Lket{M}$ denotes the  vector given as the row-wise concatenation of a matrix $M$.
As usual, for one spatial dimension (for possible generalizations to higher dimensions
see \app{sm:2d} ) we define the operators
$\Hodd$ and $\Heven$ by splitting the Hamiltonian $H = \sum_{l=1}^N H^{[l,l+1]}$ into two sums, one containing the even $H^{[2l,2l+1]}$ and one containing the odd interactions $H^{[2l-1,2l]}$.
Hence, both consist of mutually commuting terms.
If the Lindblad operators $L_\alpha$ act only on site (the case of two-site Lindblad operators is treated later on), we can
approximate $\e^{\tau \mathcal L}$ using a second order Trotter-Suzuki formula as (see the \app{supp:trac_cert} for a full error analysis)
\begin{equation} \label{eq:secorder}
\e^{\tau\mathcal{L}} = \e^{\tau \Aodd /2} \e^{\tau \Aeven /2} \e^{\tau \mathcal{D}} \e^{\tau \Aeven /2} \e^{\tau \Aodd /2} + O(\tau^3)\;,
\end{equation}
partially shown in Fig.~\ref{fig:sketchy}(d),
where $\Aeo = -\i \Heo \otimes \1 + \i \1 \otimes \bar{H}_{\nu}$ with $\nu = o,e$.
The layers $\Aeven$ and $\Aodd$ implement the coherent part of the evolution and are
identical to the usual TEBD layers.
Expressing $\rho_t$  as $\rho_t = X_t X_t^\dagger$ we find that
by setting $X_t' = \e^{-\i \tau \Aeo/2} X_{t}$ we recover exactly{$\rho_t' = \e^{\tau \Aeo /2} \rho_t$.
Hence, on the level of the local tensors $\locT$ we can just adapt the usual TEBD algorithm for nearest neighbor Hamiltonians,
to perform the coherent part of the dynamics.

The dissipative layer requires a more careful treatment
 and we exploit the fact that since the  operators $L^{[l]}$ act only on a single
site, we find $\e^{\tau \mathcal{D}} = \bigotimes_l \e^{\tau \mathcal{D}^{[l]}}$, with
\begin{equation}
\mathcal{D}^{[l]} = \sum_{j} \bigl( L^{[l]}_j \otimes \bar L^{[l]}_j - \tfrac{1}{2}
 (L^{[l]\dagger}_j L^{[l]}_j \otimes \1 + \1 \otimes L^{[l]T}_j \bar L^{[l]}_j)\bigr)\,,
\end{equation}
 where the sum runs over all Lindblad operators $L^{[l]}_j$ acting on site $l$.
 Since $\e^{\tau \mathcal{D}_l}$ is completely positive,  we can find  a set of Kraus-operators $\{B_{l,q}\}$ satisfying $\e^{\tau \mathcal{D}_l} =
\sum_{q=1}^{k} B_{l,q} \otimes \bar{B}_{l,q}$.  The action of $\e^{\tau\mathcal{D}_l}$ on the level of the local tensors is now given by a contraction of $B_{l,q}$ into $\locT_t$,
while joining the variational Kraus dimension $K$ with the Kraus rank $k\leq d^2$ of the quantum
channel $\e^{\tau \mathcal{D}_l}$, as shown in Fig.~\ref{fig:sketchy}(b).
The application of each Trotter-Suzuki  layer increases only the dimension of a single  index
of the local tensors $\locT$: The bond dimension $D$ is increased by the coherent layers, the
Kraus dimension $K$ by the dissipative layers.
This allows for immediate compression of the enlarged dimension similar to the standard DMRG-setting.
In all compression steps the Frobenius norm error introduced on the purification operators can be kept track of. This translates into a trace-norm error for the state itself. By taking also the error from the Trotter-Suzuki approximation into account, we provide an explicit bound for the trace-norm error, see Theorem~$1$ in \app{supp:trac_cert}.

The algorithm yields an overall computational cost scaling as $O(d^5D^3K) + O(d^5D^2K^2)$,
by executing a clever contraction of the coherent terms.
Moreover, the locally purified tensor network takes advantage of the gauge freedom,
e.g., by reducing costs for local measurements from $O(N)$ to $O(1)$, with $N$ being the system size.

%%%%%%%%%%%%%%%%%%%%%%%%%%%%%%%%%%%%%%%%%%%%%%%%%%%%%%%%%%%%%%%%%%%%%%%%%%%%%%%%
% PHYSICAL SYSTEMS
%%%%%%%%%%%%%%%%%%%%%%%%%%%%%%%%%%%%%%%%%%%%%%%%%%%%%%%%%%%%%%%%%%%%%%%%%%%%%%%%
{\it Numerical results.} In order to verify the applicability of our method we apply it to three prototypical benchmark situations. The first one comprises a few-body scenario, consisting of two qubits coupled via cavities with additional excitation-losses \cite{Schmidt2010}.
As genuine many body example we study the steady state of an XXZ spin $1/2$ chain with edge dissipation channels,
which allows for comparison with analytical solutions derived in Ref.~\cite{ProsenDissipativeStationaryXXZ}.
Finally we show the validity of the two-site Lindblad-operator approach in the
case of the  Kitaev wire \cite{DissipationTopologyZoller}.

% S-C-C-S
% ------------------------------------------------------------------------------

\begin{figure}
 \begin{center}
 \begin{overpic}[width = .75\columnwidth, unit=1pt]{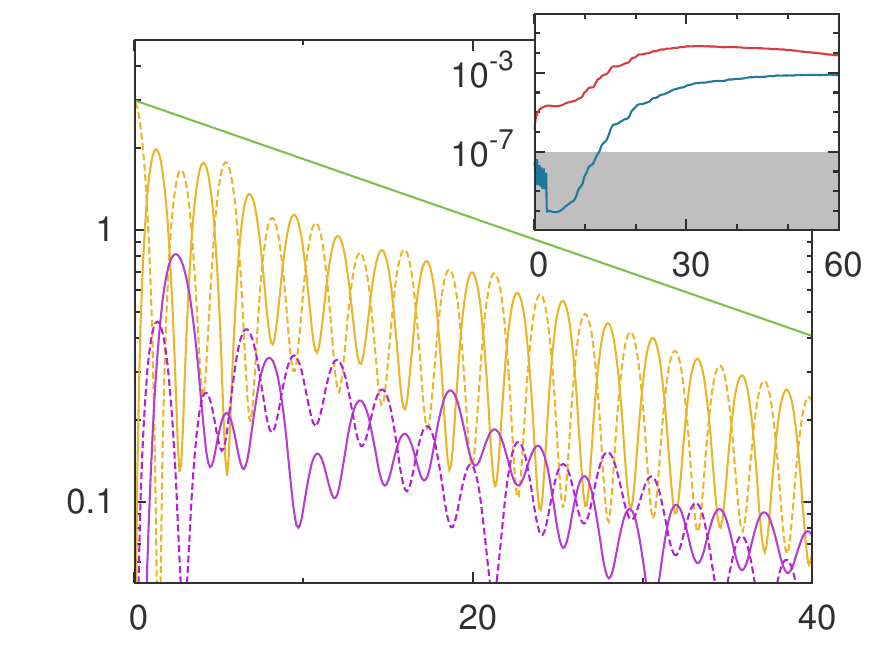}
  \put(2, 65){\large $n_j(t)$}
  \put(73, 2){\large $t$}
  \put(87, 43){\large $t$}
  \put(56, 72){\large $\Infid$}
 \end{overpic}
 \end{center}
\caption{ \label{fig:SCCS}
Main: Excitation populations of the four sites (see main text) in the coupled spin-cavity model
(yellow: cavity, violet: spin, dashed: right, straight: left),
here for $\gamma = 0.05$,
$\alpha_1 = \alpha_2 = 0.48,$ $\alpha_{CC} = -1.0,$ $\omega_{C}  = \omega_S = 1.0$,
as well as their sum $\mathcal{N}$ (green line). The latter is nicely fitted by an exponential, with
decay rate $\gamma_{\text{fit}} = 0.04997 \pm 8\cdot 10^{-5}$.
Inset: Comparison of the locally purified evolution,
here for bond dimension $D = 40$ and Kraus dimension $K = 40$,
with the exact Liouville evolution: Infidelity $\Infid$ (blue line)
and relative Hilbert-Schmidt distance (red line). Infidelities are estimated
to be numerically reliable above $10^{-7}$ (non-shaded area).
}
\end{figure}
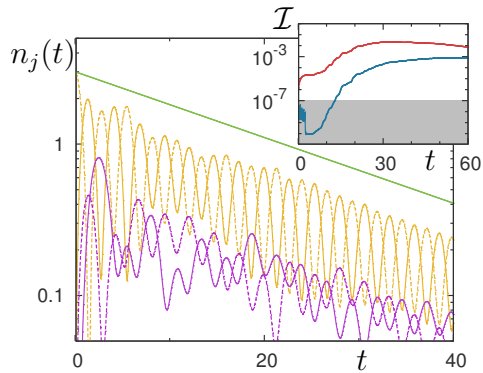

% XXZ chain
% ------------------------------------------------------------------------------

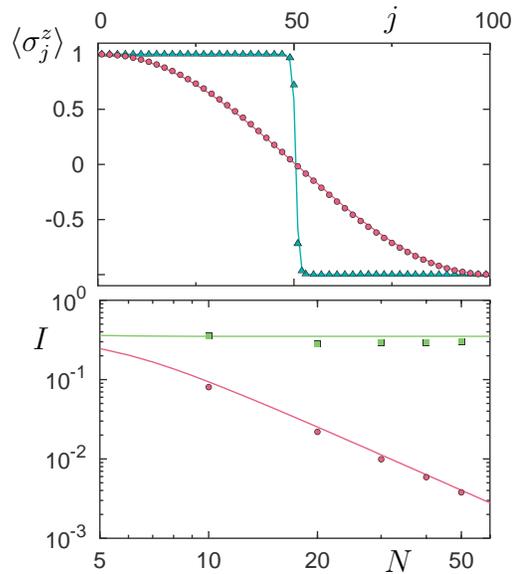
\begin{figure}
 \begin{center}
  \begin{overpic}[width = .8\columnwidth, unit=1pt]{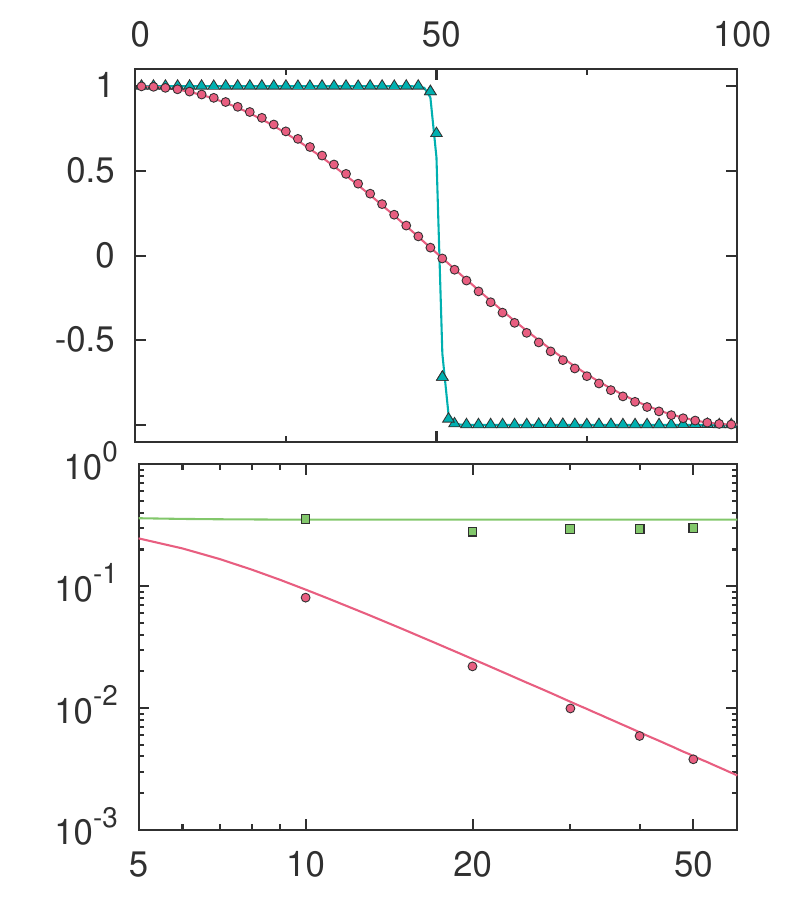}
  \put(0, 91){\large $\langle \sigma^z_j \rangle$}
  \put(64, 95){\large $j$}
  \put(64, 2){\large $N$}
  \put(4, 40){\large $I$}
 \end{overpic}
 \end{center}
\caption{ \label{fig:XXZ}
Comparison of simulated steady state (points) with analytical results (lines)
from Ref.\ \cite{ProsenDissipativeStationaryXXZ} for the XXZ model
with edge driving of $\gamma=1$ in several parameter regimes.
Respectively: green squares $\Delta = 0.5$, red dots $\Delta = 1.0$, cyan triangles $\Delta = 1.5$.
Top:
Local magnetisation in the $z$-direction $\langle \sigma^z_j \rangle$ as a function
of the site $j$, for a chain of length $N = 100$.
Bottom:
Time averaged steady state spin current $I_j = 2\, \imag \langle \sigma^{+}_{j} \sigma^{-}_{j+1}\rangle$
at the chain center $j=N/2$,
as a function of the total chain length $N$.
}
\end{figure}

In the first model, two interacting optical cavities ($C_1$ and $C_2$) are each coupled
to a private qubit ($S_1$ and $S_2$). Ordering the sites as $(S_1-C_1)-(C_2-S_2)$ results in a
nearest neighbor model to which we can apply our algorithm.
The coherent part of the dynamics is captured by the Jaynes-Cummings Hamiltonian, which describes each spin-cavity
interaction, plus a photon tunnelling between cavities.
In terms of the spin-operators $\sigma^\pm_l = (\sigma^x_l\pm i\sigma_l^y)/2$ and
the creation $c_l^\dagger$ and annihilation operators $c_l$ of the
cavity photons. The Hamiltonian is given by
\begin{equation}
H = \sum_{l = 1,2} (\alpha_{l} (\sigma^+_l c_l + \sigma^-_l c^{\dagger}_l) + \omega_{C} n_{l} + \omega_{S} \sigma^{z}_{l})
+ \alpha_{CC} ( c^{\dagger}_1 c_2 + c^{\dagger}_2 c_1 )\,.
\end{equation}
The dissipation models a homogeneous probability of excitation losses and
given by single-site Lindblad operators: $L_{S_l} = \sqrt{\gamma} \,\sigma^{-}_l$
for the spins and
$L_{C_l} = \sqrt{\gamma} \,c_l$ for the cavities.
We start the evolution in a pure product state, where only the right
cavity is nonempty and filled exactly with $N(0) = 3$ photons.
The symmetries of the model imply an easy relation between the total excitation number
$\mathcal{N}(t)$ and the coupling strength $\gamma$: $\mathcal{N}(t) = \mathcal{N}(0) \,\mathrm{e}^{- \gamma  t}$, a behaviour which is nicely captured by our simulations with high precision
(under 0.2\% relative deviation).
Fig.~\ref{fig:SCCS} shows the occupation on each site as well as $\mathcal{N}(t)$,
which correctly reproduces the expected exponential decay.
The inset shows
the infidelity
\begin{equation}
\Infid (\rho_1, \rho_2) = 1 - \Tr\left[ ({ {\rho_2}^{1/2} \rho_1 {\rho_2}^{1/2}} )^{1/2} \right]
\end{equation}
 of the locally purified evolution $\rho_1$
with respect to the exact evolution $\rho_2$ carried out in Liouville space. As expected,  close to the steady state the deviation of the
locally purified dynamics from the exact evolution
converges to a finite value (depending both on $\tau$ and on maximal bond  $D$ and Kraus dimension $K$).

As second benchmark, we consider the evolution of a spin-chain under the XXZ-Hamiltonian
\begin{equation}
H = \sum_l (\sigma^x_l \sigma^x_{l+1} +  \sigma^y_l \sigma^y_{l+1} + \Delta \sigma^z_l \sigma^z_{l+1})\,.
\end{equation}
Via Jordan-Wigner transformation this system is mapped into a spin-less fermion Hubbard model with a density-density
nearest neighbor interaction.
Therefore, with the addition of two reservoirs
embodied by Markov channels at the edges, it models
fermionic DC transport in a quantum wire.
We introduce Lindblad operators $L_{\leftarrow} = \sqrt{2 \gamma}\, \sigma^{+}_1$ at the leftmost site (source)
and $L_{\rightarrow} = \sqrt{2 \gamma}\, \sigma^{-}_N$ at the rightmost {site} (drain).
We search the steady state for different parameter {regimes $(\Delta, \gamma)$}. In order to compare our simulations with analytic results from Ref.~\cite{ProsenDissipativeStationaryXXZ}, we consider two observables, the local z-axis magnetisation $\sigma_l^z$ and the spin-current operator $I_l = i\left(\sigma^+_l\sigma^-_{l+1} - \sigma^-_l\sigma^+_{l+1}\right)$.
The steady state regime is achieved by evolving the system until the considered observables become stationary.
Fig.~\ref{fig:XXZ} shows typical results: the local magnetisation of a  chain of $N=100$ spins
is plotted in the top panel, while the current as a function of the chain length is plotted on the bottom frame.
A remarkable quantitative match to Ref.~\cite{ProsenDissipativeStationaryXXZ} emerges even for small
$D$ and $K$ ($\sim 60$).

% 2-site Lindblad operators
% ------------------------------------------------------------------------------

\begin{figure}
 \begin{center}
 \begin{overpic}[width = .75\columnwidth, unit=1pt]{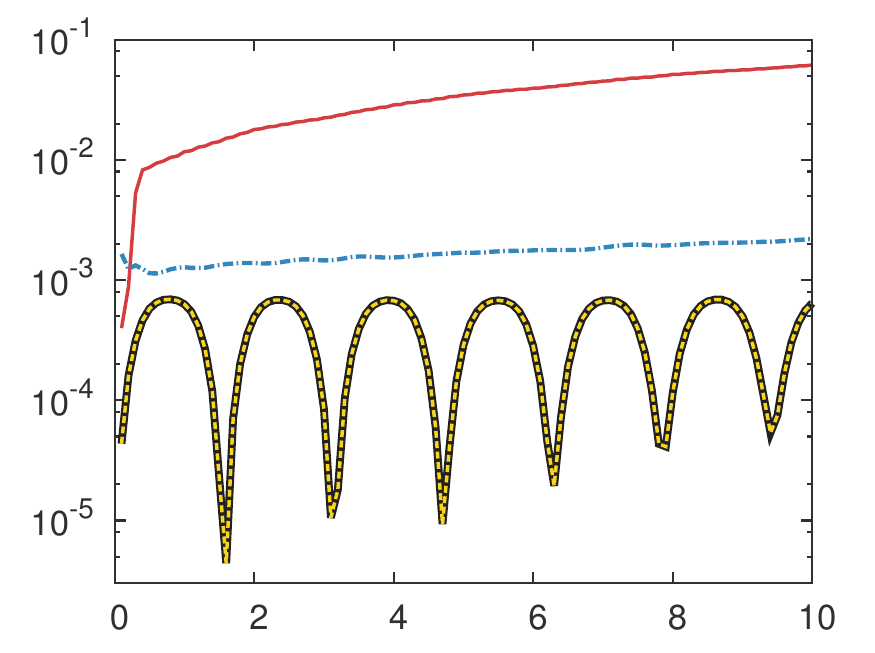}
  \put(0, 62){\large $\Infid$}
  \put(85, 2){\large $t$}
 \end{overpic}
 \end{center}
\caption{ \label{fig:Kita}
Infidelities $\Infid$ of the two-site Lindblad evolution strategies for the Kitaev wire model,
as a function of time $t$, for $N = 6$ sites.
Strategies \emph{a)}, \emph{b)} and \emph{c)} are shown (thick black line overlapping with the yellow dashed line, red line, and yellow dashed line respectively).
The cyan dot-dashed line shows the error estimator calculated via truncated singular value data, for strategy \emph{c)}.
It is obtained via Eq.~\eqref{eq:estimator} of the \app{supp:trac_cert}.
Here we used $K = D = 30$.
}
\end{figure}

Finally, we apply our method in a generalized setting, where the Lindblad operators are two-local instead of on-site
(see the \app{app:nn_ch} for details).
We employ a two-terms Trotter-Suzuki  approximation by odd
$\mathcal{L}_o = \Aodd + \Dodd$ and even $\mathcal{L}_e = \Aeven + \Deven$ terms.
After computing the Kraus-decomposition for the corresponding nearest-neighbor channels
\begin{equation}
\e^{\tau \mathcal{L}_{2l,2l+1}} = \sum_q^k B^{[2l,2l+1]}_q\otimes\bar{B}^{[2l,2l+1]}_q
\end{equation}
one can choose how to implement the action of $B^{[2l,2l+1]}_q$ onto $A^{[2l]}$ and $A^{[2l+1]}$.
In particular, there are different
possibilities for the distribution of the Kraus rank $k$ of the channel between the Kraus dimensions $K_{2l}$ and $K_{2l+1}$ of the two sites.
Moreover, when such dimension $k$ is distributed nontrivially ($k_1 > 1$ to the left site, and $k_2 > 1$ to the right site,
where $k_1 k_2 = k$) there is an additional freedom, represented by a unitary transformation $U$ in the $k$-dimensional auxiliary
space, that will influence the precision of the whole algorithm. This gauge transformation $U$
is discussed in detail in \app{sec:method}, alongside a numerical technique we adopt to optimize it.
For appropriate comparison, we consider three strategies: \emph{a)} Kraus rank all to one side ({\it e.g.}~$k_2 = 1$),
\emph{b)} Kraus rank distributed as evenly as possible ($k_1 \simeq k_2 \simeq \sqrt{k}$),
with random $U$ (unoptimized),
\emph{c)} analogously to \emph{b)} but with optimized $U$.

As a testbed we consider a Kitaev wire with two-site Lindblad operators as presented in
Ref.\ \cite{DissipationTopologyZoller} and compare it again with the exact numerical evolution to study fidelities.
The system is a spinless fermion chain, with Hamiltonian
\begin{equation}
H = \sum_{j} ((-J \,a^{\dagger}_{j} a_{j+1} - \Delta \,a_{j} a_{j+1} + \mbox{h.c.}) - \mu \,a^{\dagger}_{j} a_{j})\,,
\end{equation}
and a single Lindblad channel for every pair of neighboring sites
$L_{j,j+1} = \sqrt{\gamma} (a^{\dagger}_{j} + a^{\dagger}_{j+1})(a_{j} - a_{j+1}) / 4$.
The parameters are chosen in order to target the non-trivial topological phase \cite{DissipationTopologyZoller}: $J=\Delta=1.0$,
$\mu = 0.0$, and $\gamma = 10^{-2}$.
Fig.\ \ref{fig:Kita} shows that we can capture accurately the real-time evolution starting from
an entangled mixed random state, by direct comparison of our scheme with the exact Liouville evolution,
for a chain of $6$ sites.
It also suggests that strategies \emph{a)} and
\emph{c)} yield, surprisingly, equivalent precision, and are preferable choices than strategy \emph{b)}.
Moreover, we plot the error estimator obtained from the discarded singular values during state compression (discarded weight, see Definition \ref{def:dw}  \app{supp:trac_cert}) for strategy \emph{c)}, and show that it is an effective
upper bound to the actual infidelity error.

%%%%%%%%%%%%%%%%%%%%%%%%%%%%%%%%%%%%%%%%%%%%%%%%%%%%%%%%%%%%%%%%%%%%%%%%%%%%%%%%
% PERSPECTIVES
%%%%%%%%%%%%%%%%%%%%%%%%%%%%%%%%%%%%%%%%%%%%%%%%%%%%%%%%%%%%%%%%%%%%%%%%%%%%%%%%

{\it Perspectives.} In this work, we have introduced a versatile algorithm for simulating open quantum many-body systems. All errors made by the algorithm are bounded in trace-norm.
The ideas presented here
overcome a number of previous limitations and allow us to probe both transient behaviour and the full dynamics to stationary states.
We have discussed three important
benchmark cases, and a number of perspectives open up here:
Firstly, the framework can be used to analyse weakly interacting open quantum systems, perturbing frequently studied free fermionic
models to study topology generated by engineered dissipation \cite{DissipationTopologyZoller,TopologyDissipationLong}. Clearly, notions of algebraic and exponential dissipation can readily be accessed \cite{CaiBarthel}, as well as glass-like dynamics \cite{GlassyKollath} and kinematic inhibitance, or the interplay between localization by dissipation and disorder.
Furthermore, the method finds immediate application in dissipative quantum engineering of entangled many-body states \cite{Dissipengineer},
for instance by merging with optimal control techniques \cite{DoriaControl}.
It will also be an invaluable
tool in exploring shortcuts to adiabaticity \cite{Shortcuts} in open-system quantum many-body settings.
Another intriguing enterprise is to further explore questions of stability of stationary states under local Liouvillian perturbations, in particular without the assumption of
a finite log-Sobolev constant or rapid mixing \cite{Mixing,MadridRobustness}. It would be also exciting to explore formulations
of our method in a time-dependent variational principle framework \cite{Lagrangian,Flow}.

{\it Acknowledgements.} This work has been supported by
the EU (SIQS, RAQUEL, COST, AQuS),
the ERC (TAQ),
and the DFG (SFB/TRR21).
We thank T.\ Prosen for discussions. We acknowledge the BWgrid for computational resources.

%\bibliography{BigReferences9}
%

\section*{Appendix}
In this appendix we provide some additional details on our numerical method and prove the trace norm certificate for the algorithm. In section \ref{supp:1} we review some definitions well known facts and relations between different matrix norms. In section \ref{supp:trotter_bounds} we we extend the Trotter-Suzuki approximation for  Liouvillian time-evolution with time-constant generators in terms of the diamond norm from \cite{KliBarGog11} to second order before we provide additional details on our numerical method in section \ref{sec:method}. The full error analysis of our algorithm and the trace norm certificate are presented in section \ref{supp:trac_cert}.

\subsection{Notation and simple facts}\label{supp:1}
In order to do make our error bounds precise we need to find meaningful error measures.
The optimal distinguishability of observables by single expectation values is the operator norm $\norm{\argdot}_\infty$, which coincides with the largest singular value. The optimal distinguishability of quantum states by single expectation values is given in terms of the trace-norm defined by
\begin{equation}
 \norm{R}_1 \coloneqq \Tr( \sqrt{R^\dagger R}) \, .
\end{equation}
Similarly, we define the Frobenius norm by
\begin{equation}
 \norm{R}_2 \coloneqq \sqrt{\Scp{R}{R}} \, ,
\end{equation}
where $\Scp{Q}{R} \coloneqq \Tr(Q^\dagger R)$ denotes the Hilbert-Schmidt inner product. Note that
\begin{equation}\label{eq:2norm-1norm}
 \norm{R}_2\leq \norm{R}_1 \leq \sqrt{\rank(R)} \norm{R}_2
\end{equation}
for any operator $R$. Importantly, the Frobenius can be defined in a similar way for all types of tensors and coincides with the usual vector $2$-norm of the vectorized tensors.

A notion of closeness of two quantum state $\rho$ and $\sigma$ is given by the \emph{fidelity}
\begin{equation}
 \fid(\rho, \sigma)
 \coloneqq
 \Tr\Bigl( \sqrt{ \sqrt{\sigma} \rho \sqrt{\sigma}}\, \Bigr) \, .
\end{equation}
Fidelity and trace norm are related as
\begin{equation} \label{eq:fidelity_trace_norm}
  1-\fid(\rho,\sigma)\leq \kw 2 \norm{\rho-\sigma}_1\leq \left({1-\fid^2(\rho,\sigma)}\right)^{1/2}\;.
\end{equation}
The operational distinguishability of quantum channels is given by the diamond norm, which is defined by
\begin{align}
\norm{T}_\diamond \coloneqq \sup_{n\geq 1} \norm{T\otimes \id_n}_{1\to 1}
\end{align}
for any linear map $T$ on operators (e.g., a difference of quantum channels), where
\begin{equation}
\norm{T}_{1\to 1} \coloneqq \sup_{\norm{R}_1=1} \norm{T(R)}_1
\end{equation}
defines the ($1\to 1$)-norm and $\id_n$ denotes the identity channel acting on an $n$-dimensional quantum system. In fact, the supremum over $n$ is attained for $\id_n$ being the identity acting on a copy of the system $T$ acts on. The ($1\to 1$)-norm measures the distinguishability of quantum channels given by applying the channels to a states and measuring an observable and optimizing over all states and observables. The diamond norm measures the distinguishability of quantum channels in a similar way but where arbitrary ancilla systems are allowed. Fortunately, the diamond norm can be calculated efficiently (in the Hilbert space dimension) \cite{Wat12}. Moreover, the diamond of a arbitrary Liouvillian \eqref{def:dw} is bounded as
\begin{align}
 \dnorm{\L} \leq 2 \norm{H}_\infty + 2 \sum_{\alpha} \norm{L_\alpha}^2_\infty \, .
\end{align}

In the end, we would like to have a bound on the fidelity or trace-norm error of the state we simulate with our algorithm. The algorithm works in a purification vector space, where our compression steps introduce $2$-norm errors. Fortunately, they translate into fidelity or trace-norm errors for the actual states. For the case without any locality structure, this is a well-known fact captured by the following lemma. For sake of completeness, we will also provide a proof.

\begin{lemma}[Trace-norm and fidelity bound]\label{lem:2normdist}
	Let $\rho = X X^\dagger$ and $\sigma=Y Y^\dagger$ be density operators. Then
	\begin{align}
	\norm{\rho - \sigma}_1 &\leq \sqrt{2} \norm{X-Y}_2 \label{eq:Fnorm2trnorm}
	\intertext{and}
	\fid(\rho, \sigma) &\geq \frac{1}{2}\left(2-\norm{X-Y}_2^2\right) . \label{eq:Fnorm2fid}
	\end{align}
\end{lemma}

\begin{proof}[Proof of bound~\eqref{eq:Fnorm2trnorm}]
	We view $X$ as a purification of $\rho$ and denote by $\ketbra XX$ the corresponding density matrix, and similarly for the purification $Y$ of $\sigma$.
	In fact, we show a stronger statement:
	\begin{equation}\label{eq:stronger}
	\norm{\ketbra XX - \ketbra YY }_1 \leq \sqrt{2} \norm{X-Y}_2 \, .
	\end{equation}
	$\sigma$ and $\rho$ are partial traces of $\ketbra XX$ and $\ketbra YY$, respectively. The proposition then follows from the fact that the trace distance is non-increasing under partial traces.
	
	To prove Eq.~\eqref{eq:stronger} we use the inequality \eqref{eq:2norm-1norm}, the fact that the $2$-norm comes from the Hilbert-Schmidt inner product, the normalisations of $\rho$ and $\sigma$, and that this normalization implies $\norm{X}_2 = \norm{Y}_2 = 1$ to write
	\begin{align}
	\norm{\ketbra XX - \ketbra YY}_1^2
	&\leq
	2\norm{\ketbra XX - \ketbra YY}_2^2
	\\
	&=
	2\left(2-2 \Re \Tr\bigl(\ketbra X X^\dagger \, \ketbra Y Y\bigr)\right)\nonumber
	\\
	&=
	2\left(2-2|\braket X Y |^2\right)\nonumber
	\\
	&\leq 2 \norm{X-Y}_2^2 \, .	\nonumber
	\end{align}
	This finishes the proof.
\end{proof}

\begin{proof}[Proof of bound~\eqref{eq:Fnorm2fid}]
Using the decomposition $\rho=X^\dagger X$ and $\sigma = Y^\dagger Y$ and the definition $\Abs{A} = \sqrt{A^\dagger A}$ we obtain
\begin{align}\label{eq:fid_I}
\fid(\rho,\sigma) &= \Tr\Abs{\left(\Abs{X}\cdot\Abs{Y^\dagger}\right)}\geq \Abs{\Tr\left(X^\dagger Y\right)}\\
&\geq \frac{1}{2}\left( \Scp{X}{Y}+ \Scp{Y}{X}\right) ,\nonumber
\end{align}
where $\Scp{X}{Y}$ denotes the Hilbert-Schmidt-scalar product of $X$ and $Y$. Note that  $\Tr\rho=1$ implies $\norm{X}_2^2=1$.
Hence, inserting into \eqref{eq:fid_I} the canonical decomposition of the $2$-norm into overlaps, given by
\begin{align}\label{eq:2nscp}
  \norm{X-Y}^2_2 = \norm{X}^2 + \norm{Y}^2 - \Scp{X}{Y} - \Scp{Y}{X}
\end{align}
completes the proof.
\end{proof}

\subsection{Explicit error bounds for dissipative Trotter-Suzuki approximations}\label{supp:trotter_bounds}
In this section we derive and discuss error bounds on Trotter-Suzuki approximations for Markovian dynamics. We build on the derivation of a first order Trotter-Suzuki approximation from Ref.~\cite{KliBarGog11} to prove an error bound also for the second order. The results in Ref.~\cite{KliBarGog11} allow for explicitly time dependent Liouvillians. But here, we restrict our analysis to the time-constant case.

Often, error bounds for Trotter-Suzuki approximations only give the scaling in the time step but not in the system size \cite{Suz91}. A subtle point in the derivation of such bounds is that even though quantum channels generated by Liouvillians can be inverted as linear maps but their inverses are in general no quantum channels and are not normalized in diamond norm.

The following bound is an important step in the proof of the main theorem in Ref.~\cite{KliBarGog11} and will also help us to establish the second order Trotter-Suzuki approximation.

\begin{lemma}[Part of Theorem~5 in SM of Ref.~\cite{KliBarGog11}]
	\label{lem:commutator_T_K}
	Let $\K$ and $\L$ be Liouvillians and $r\geq 0$. Then
	\begin{align}
	\dnorm{[\e^{r\L}, \K]} \leq r\, \e^{r\dnorm{\L}} \dnorm{[\K,\L]} \, .
	\end{align}
\end{lemma}

\begin{proof}
	Importantly, the diamond norm is sub-multiplicative. Writing the exponential as a power series, using the triangle inequality, and supmultiplicativity, it follows that $\dnorm{\e^{- s \L}} \leq \e^{s \dnorm{\L}}$ for all $s \geq 0$.
	
	Now, taking the norm of the identity
	\begin{align}
	[\e^{-r\L}, \K]
	&=
	\left(\e^{-r\L} \K \e^{r\L} - \K \right) \e^{-r \L}
	\\&=
	\int_0^r \intd u \, \frac{\partial}{\partial u} \left( \e^{-u\L} \K \e^{u\L}\right)\e^{-r \L}\nonumber
	%\\&=
	%\int_0^r \intd u \,  \e^{-u \L} [\K,\L] \e^{u \L} \e^{-r \L}
	\\&=
	\int_0^r \intd u \,  \e^{-u \L} [\K,\L] \e^{-(r-u) \L} \, ,\nonumber
	\end{align}
	using the triangle inequality, sub-multiplicativity, and bounding the exponentials proves the lemma.
\end{proof}

An equation from the proof of Theorem~5 in Ref.~\cite{KliBarGog11} for our setting reads as
\begin{multline}\label{eq:first_order_calc}
\e^{t(\K+\L)} - \e^{t\K} \e^{t \L}
\\=
\int_0^t\int_0^r \intd r \,\intd u \, \e^{\K t} \e^{\L t} \e^{-u \K} [\L, \K] \e^{-(r-u) \K} \e^{r(\K+\L)}.
\end{multline}
This yields a first order Trotter-Suzuki approximation:

\begin{lemma}[Dissipative Trotter-Suzuki product formula \cite{KliBarGog11}]\label{lem:first_order}
	Let $\K$ and $\L$ be Liouvillians and $0\leq t \, \dnorm{\K}\leq 1$. Then
	\begin{equation}
	\dnorm{\e^{t(\K+\L)} - \e^{t \K} \e^{t\L}}
	\leq
	t^2 \dnorm{[\K,\L]} \, .
	\end{equation}
\end{lemma}

\begin{proof}
Taking the norm of Eq.~\eqref{eq:first_order_calc}, applying the triangle inequality, and using sub-multiplicativity together with the inequality
$\dnorm{\e^{r \L}} =1$ for $r\geq 0$ yields with $a\coloneqq \dnorm{\K}$
\begin{align}
	\dnorm{\e^{t(\K+\L)} - \e^{t \K} \e^{t\L}}
	&\leq \frac{1- \e^{at}(1-at)}{a^2}\dnorm{[\K,\L]}
\end{align}
Using that $1+\tau+\tau^2-\e^\tau \geq 0$ for all $\tau \leq 1$ with $\tau = at$ finishes the proof.
\end{proof}

Now we are ready to go to the next order.

\begin{lemma}[Dissipative 2nd oder Trotter-Suzuki product formula] \label{lem:second_order}
	Let $\K$ and $\L$ be Liouvillians and $0 \leq t(\dnorm{\L} + \dnorm{\K}/2) \leq 1$. Then
	\begin{align}
	 \dnorm{\e^{t(\K+\L)} - \e^{\frac t 2 \K} \e^{t\L} \e^{\frac t 2 \K} }
	 \leq
	 \frac{t^3}{3} \dnorm{[\K,\L]} \left(\dnorm{\L} + \kw2 \dnorm{\K} \right)\, .
	\end{align}
\end{lemma}
This lemma provides a new error bound that might be of independent interest. It remains an open problem to also find higher order Trotter-Suzuki approximations for Markovian dynamics. The crucial point is, that unlike to unitary dynamics, one cannot evolve backwards in time and obtain a completely positive time evolution in general.

\begin{widetext}
\begin{proof}
	Quantum channels are normalized in diamond norm, hence $\dnorm{\e^{r \L}} =1$ for $r\geq 0$.
	From the identity
	\begin{align}
	\e^{-\tfrac{t}{2} \K } \e^{-t \L} \e^{-\tfrac{t}{2} \K } \e^{t(\K+\L)} - \id
	&=
	\int_0^t \intd s \frac{\partial}{\partial s} \left(
		\e^{-\tfrac{s}{2} \K } \e^{-s \L} \e^{-\tfrac{s}{2} \K } \e^{s(\K+\L)}
	\right)
	\\ &=
	-\int_0^t \intd s \, \e^{-\frac{s}{2} \K}
	\left[\tkw 2 \K + \L, \e^{-s\L}\e^{-\frac{s}{2} \K}\right] \e^{s(\K+\L)}\nonumber
	\\ &=
	-\int_0^t \intd s \,\int_0^s \intd r \, \e^{-\frac{s}{2} \K}
		\left[\tkw 2 \K + \L,\frac{\partial}{\partial r}\left(\e^{-r\L}\e^{-\frac{r}{2} \K}\right)\right] \e^{s(\K+\L)}\nonumber
	\\ &=
	\int_0^t \intd s \,\int_0^s \intd r \, \e^{-\frac{s}{2} \K}
		\left[\tkw 2 \K + \L,\e^{-r\L}(\tkw 2 \K + \L)\e^{-\frac{r}{2} \K}\right] \e^{s(\K+\L)}\nonumber
	\\ &=
	\int_0^t \intd s \,\int_0^s \intd r \, \e^{-\frac{s}{2} \K}\nonumber
	\left(
	\tkw 2 [\K,\e^{-r\L}] (\tkw 2 \K + \L) \e^{-\frac{r}{2} \K}
	+\e^{-r \L} (\tkw 2 \K + \L) [\L,\e^{-\frac r 2 \K}]
	\right) \e^{s(\K+\L)}\nonumber
	\end{align}
	and using Lemma~\ref{lem:commutator_T_K} we conclude that
	\begin{align}
	\dnormb{\e^{t(\K+\L)} - & \e^{\frac t 2 \K} \e^{t\L} \e^{\frac t 2 \K} }
	\\ \nonumber&\leq
	\int_0^t \intd s \,\int_0^s \intd r
	\dnorm{\e^{\frac{t-s}{2} \K}
		\left(
		\tkw 2 [\K,\e^{-r\L}] (\tkw 2 \K + \L) \e^{-\frac{r}{2} \K}
		+\e^{-r \L} (\tkw 2 \K + \L) [\L,\e^{-\frac r 2 \K}]
		\right) \e^{s(\K+\L)}}
	\\ \nonumber&\leq
	\int_0^t \intd s \,\int_0^s \intd r
	\dnorm{
		\tkw 2 [\K,\e^{-r\L}] (\tkw 2 \K + \L) \e^{-\frac{r}{2} \K}
		+\e^{-r \L} (\tkw 2 \K + \L) [\L,\e^{-\frac r 2 \K}]
		}
	\\ \nonumber&\leq
	\int_0^t \intd s \,\int_0^s \intd r \left(
	\tkw 2 \dnorm{[\K,\e^{-r\L}]} (\tkw 2\dnorm{\K} + \dnorm{\L}) \e^{\frac{r}{2} \dnorm{\K}}
		+\e^{r \dnorm{\L}} (\tkw 2 \dnorm{\K} + \dnorm{\L}) \dnorm{[\L,\e^{-\frac r 2 \K}]}
	\right)
	\\ \nonumber&\leq \left(\dnorm{\L} + \kw2 \dnorm{\K}\right)
	\int_0^t \intd s \,\int_0^s \intd r \left(
	\frac{r}{2} \e^{r \dnorm{\L}}\dnorm{[\K,\L]} \e^{\frac{r}{2} \dnorm{\K}}
	+\frac{r}{2} \, \e^{r \dnorm{\L}}  \dnorm{[\L,\K]} \e^{\frac r 2 \dnorm{\K}}
	\right)
	\\ \nonumber&=
	a\int_0^t \intd s \,\int_0^s \intd r \,
	r \,  \e^{r a}  \dnorm{[\K,\L]}
	\\ \nonumber&=
	\frac{2(1-\e^{a t}) + a t (1+\e^{a t})}{a^2} \dnorm{[\K,\L]}
	\end{align}
	where $a \coloneqq \dnorm{\L} + \kw2 \dnorm{\K}$.
	The fact that $2(1-\e^\tau) + \tau (1+\e^\tau) \leq \tau^3/3$ for all $\tau \in [0,1]$ with $\tau = at$ finishes the proof.
\end{proof}
\end{widetext}
Using first order Trotter-Suzuki approximation (Lemma~\ref{lem:first_order}) for the simulation of the dynamics of a spin chain of length $N$ with $m=t/\Delta t$ times steps yields a total Trotter-Suzuki error scaling as
\begin{align}
\text{TS-error}_1
\in
\landauO \left( a^2 N t^2/m \right) \, ,
\end{align}
where the local Liouvillian terms are bounded by $a$, i.e., $\dnorm{\L_{j,j+1}} \leq a$.
Using the second order Trotter-Suzuki approximation (Lemma~\ref{lem:second_order}) instead, yields (see Eq.~\eqref{eq:TZ-error_in_pf} below)
\begin{align}
\text{TS-error}_2
\in
\landauO \left( a^3 t^3 N^2 /m^2 \right) \, .
\end{align}
Hence, for a given maximum diamond norm error $\epsilon>0$, a number of local unitaries scaling as
\begin{align}
m_1 &= \landauO\left( \frac{(a\, t)^2 N }{\epsilon} \right),
\\
m_2 &= \landauO\left( \frac{(a\,t)^{3/2} N }{\sqrt{\epsilon}} \right)
\end{align}
is sufficient for the simulation.

\subsection{Simulation method expanded}\label{sec:method}
In this section we give some more specifics about our algorithm.
As described in the main text, locally purified states are taken as a variational set in the algorithm. In particular, we make use of the fact that this form can be preserved under the application of nearest-neighbor unitary operators and local channels. The same is possible if we consider quantum channels acting on neighboring spins once we preprocess and optimize the distribution of the Kraus indices (see Section \ref{app:nn_ch}).

In order to discuss the details of our algorithm let us recall some properties of MPS and MPOs \cite{Schollwock201196} which directly carry over to our setting. As in the main text, we decompose the density matrix $\rho$ as $\rho=X X^\dagger$ and consider the MPO description of $X$ in terms of a set of local tensors $\{\locT\}_{l=1}^N$. Using the notation
\begin{equation}
 [n]\coloneqq \{1,2,\dots,n\} \, ,
\end{equation}
let us also fix the following convention for the four indices of the local tensors $\locT$, namely that we will write
\begin{equation}
\locT = \left(
  \locTel{[l], s,r}{n,m}
	      \right)_{s \in [d], r \in [K], n \in [D_{l-1}], m\in [D_{l}]}
\end{equation}
for the left bond dimension $D_{l-1}$, the Kraus dimension $K$, the right bond dimension $D_{l}$ and the physical dimension $d$.
The decomposition of a state $\rho$ into local tensors $\{\locT\}$ is not unique, but allows for some gauge-freedom. In addition to the usual gauge freedom on the bond indices there is an additional gauge freedom with respect to the Kraus index.
This allows us to use singular value decompositions in order to achieve a mixed normalized form of the local tensors $\{\locT\}$ with respect to the orthogonality centre $l_0$ such that
\begin{align}
    \sum_{s,r,n} \locTel{[l], s,r}{n,m}\; \overline{ \locTel{[l],s,r}{n,m'}} &= \delta_{m,m'} ,& \text{if } l<l_0,\\
    \sum_{s,r,m} \locTel{[l], s,r}{n,m}\; \overline{ \locTel{[l],s,r}{n',m}} &= \delta_{n,n'} ,& \text{if } l> l_0\;.
\end{align}
Next we show how the local tensors $\{\locT\}$ transform under the application of a quantum channel $T$ that acts non-trivially only on the one lattice site $l$. We can find local Kraus operators $B_q = \1\otimes B_q^{[l]}\otimes\1$ such that
\begin{align}
 T(\rho)
 &=
 \sum_{q = 1}^{k} B^{[l]}_q X (B_q X)^\dagger
 \\ &=
 \left(\sum_{q=1}^{k} ( B^{[l]}_q X ) \otimes \bra{q}\right)\left(\sum_{q'=1}^{k} ( B^{[l]}_{q'} X )\otimes \bra{q'}\right)^\dagger\;\,\nonumber
\end{align}
with $k$ the Kraus rank of channel $T$.
Expanding $X$ and $B^{[l]}$ we find
\begin{align}
   &\sum_{q=1}^{k} (B^{[l]}_q X) \otimes \bra{q} =\nonumber \\
   &= \sum_{q,\vec{r},\vec{s}} \sum_{\bar{s} = 1}^{d}
   [X]^{s_1 ,\ldots,\bar{s},\ldots s_N}_{r_1 ,\ldots , r_N} B^{[l],s_l,\bar{s}}_q
   \ketbra{\vec{s}}{\vec{r}}\otimes\bra{q}\;.
\end{align}
Using the definition of $X$, we see that we can absorb the factors $B^{[l],s_l,\bar{s}}_q\bra{q}$ into the local tensor $A^{[l]}$ and set
\begin{align}
  \widetilde{A}^{[l]}_{m,n}=\sum_{q,s,\bar{s},r} A_{m,n}^{[l],\bar{s},r} B^{[l],s,\bar{s}}_q \ketbra{s}{r,q} =
  \sum_{s,r^\prime} \widetilde A_{m,n}^{[l],s,r^\prime} \ketbra{s}{r^\prime},
 \end{align}
by combining the indices $q$ and $r$ into a combined Kraus index $r^\prime$,
whose dimension $K' = k \cdot K$ is the product of the original Kraus dimension $K$ in the locally purified
tensor network, and the Kraus rank $k$ of the quantum channel $T$.
This procedure is sketched in Figure \ref{fig:sketchy}(b) (bottom panel),
and shows some similarities with the contraction scheme for projected entangled pair states discussed
in Ref.~\cite{Lubasch2014}.
Note that the non-uniqueness of the Kraus-operators decomposition directly implies the aforementioned gauge-freedom of the Kraus index
$q$ (and thus also $r'$) with respect to the unitary group.

\subsubsection{Nearest neighbor couplings}
\label{app:nn_ch}
Here we discuss in greater detail the scenario of nearest neighbor Liouvillians
\begin{equation}\label{eq:local_Liouvillian}
 \L = \sum_{l=1}^{N-1} \L^{[l,l+1]} \, ,
\end{equation}
where the local Liouvillian terms are written in Lindblad form as
\begin{equation}
 \L^{[l,l+1]} = -\i \, [H^{[l,l+1]}, \argdot ] + \sum_j L^{[l,l+1]}_j \, ,
\end{equation}
and where each Hamiltonian term $H^{[l,l+1]}$ and Lindblad operator $L^{[l,l+1]}_j$ acts only on the two neighboring sites $l$ and $l+1$ non-trivially. We split the dynamical map $\e^{\tau \mathcal{L}}$ into odd and even terms using the Trotter-Suzuki approximation from Lemma~\ref{lem:second_order}
\begin{equation}\label{eq:Trotter_even_odd}
 \e^{\tau \mathcal{L}} = \e^{\tau \mathcal{L}_o /2} \e^{\tau \mathcal{L}_e} \e^{\tau \mathcal{L}_o /2} + O(\tau^3)
\end{equation}
where
\begin{equation}\label{eq:L_even_odd}
\L_o \coloneqq  \sum_l \L^{[2l-1,2l]} \ \text{ and } \
\L_e \coloneqq \sum_l \L^{[2l,2l+1]}.
\end{equation}
By commutativity it trivially holds that
$\e^{\tau \mathcal{L}_o} = \bigotimes_{l} \e^{\tau \L^{[2l-1,2l]}}$ and
$\e^{\tau \mathcal{L}_e} = \bigotimes_{l} \e^{\tau \L^{[2l,2l+1]}}$.
With respect to row-wise vectorization of operators, the super-operator $\L^{[l,l+1]}$ takes the form
\begin{multline} \label{eq:neichan}
 \L^{[l,l+1]} = -\i\, H^{[l,l+1]} \otimes \1 + \i\, \1 \otimes \bar{H}^{[l,l+1]} + \\
 + \sum_{j} \left( L^{[l,l+1]}_{j} \otimes \bar L^{[l,l+1]}_j - \frac{1}{2}\;
 (L^{[l,l+1]\dagger}_j L^{[l,l+1]}_j)\otimes \1 + \right.  \\
 \left. - \frac{1}{2}\; \1 \otimes (L^{[l,l+1]T}_j \bar L^{[l,l+1]}_j)\right).
\end{multline}
The exponential $\e^{\tau \L^{[l,l+1]}}$ is calculated numerically, then Choi transformed \cite{Choimap}
and finally Cholesky decomposed in order to obtain its Kraus decomposition into operators
\begin{equation}
B^{[l,l+1]}_q = \left( B^{[l,l+1]s_{\rm out}, s_{\rm in}}_q \right)_{s_{\rm out}, s_{\rm in} \in [d^2]}
\end{equation}
which satisfies
\begin{equation}	
\sum_{q=1}^k B^{[l,l+1]}_q \otimes \bar B^{[l,l+1]}_q = \e^{\tau \mathcal{L}^{[l,l+1]}}
\end{equation}
 and
$\sum_{q=1}^k B^{[l,l+1]\dagger}_q B^{[l,l+1]}_q = \1$
for some Kraus rank $k \leq d^4$.

A sensitive drawback, with respect to the on-site channel scenario,
is that contracting the two-site Kraus map $B^{[l,l+1]}_q$ into the purification operator $X$
enlarges multiple tensor network links simultaneously (we remind that in the previous case only one link
was enlarged per contraction). This encourages us to identify a clever contraction-compression
scheme which can keep the errors small and the computation efficient.

Namely, in order to allow for a practical contraction of $B^{[l,l+1]}_q$ into the tensor $\locT$ and $\locT[l+1]$ such that the result is again a pair of local tensors it is useful to decompose the Kraus operators according to
\begin{equation}
 B_{q}^{[l,l+1] s^{[l]}_{\rm out}\cdot s^{[l+1]}_{\rm out} ,s^{[l]}_{\rm in}\cdot s^{[l+1]}_{\rm in}}
  =
  \sum_{m=1}^{D'} B_{\leftarrow,q_1,m}^{[l]s^{[l]}_{\rm out},s^{[l]}_{\rm in}}\; \cdot\;
  B_{\rightarrow,q_2,m}^{[l+1]s^{[l+1]}_{\rm out},s^{[l+1]}_{\rm in}} \; ,
\end{equation}
where $s^{[l]}_x \in [d]$, the dimension $D'$ plays the role of a bond dimension of the channel, and
the Kraus rank $k$ has been arbitrarily split between the two sites into $k_1$ and $k_2$
so that $k_1 \cdot k_2 = k$ (say $q_1 \in \{1 ,\dots,  k_1\}$, $q_2 \in \{1 ,\dots, k_2\}$
and $q = k_2 (q_1-1) + q_2$).
It is clear that this decomposition is not unique in many ways.
Moreover, besides the application of a gauge transformation through the $D'$ bond link,
and the freedom of spliting $k$ into $k_1$ and $k_2$,
we can contract an arbitrary unitary matrix $U$ into the Kraus index $q$ of $B_q^{[l]}$.
From a computational point of view it would be preferable to find a matrix $U$ and a distribution of the Kraus indices such that the resulting bond dimension $D'$ is minimal, because in that case the growth in bond dimension linking
$A^{[l]}$ to $A^{[l+1]}$ would be minimal.
Finding such optimal $U$ poses indeed a difficult non-linear optimization problem which might not have an efficient solution.
Notice, however, that this decomposition can be seen as a single preprocessing procedure of the Trotter-Suzuki steps of the simulation
(for a time-independent master equation).
So in this sense it comprises at most a constant (non-scaling) overhead that we have to invest before starting our algorithm.

In practice, we perform an iterative, direct search of the optimized $U$ as follows.
Given $k_1$ and $k_2$ so that $k_1 k_2 = k$, and a unitary $k \times k$ matrix $U$. We evaluate
\begin{equation}
  C_{q_1, q_2}^{s^{[l]}_{\rm out}, s^{[l+1]}_{\rm out} ,s^{[l]}_{\rm in}, s^{[l+1]}_{\rm in}} =
  \sum_{q=1}^{k}
  U^{q}_{q_1,q_2}
  \,B_q^{[l,l+1]s^{[l]}_{\rm out}, s^{[l+1]}_{\rm out} ,s^{[l]}_{\rm in}, s^{[l+1]}_{\rm in}}.
\end{equation}
Then we decompose the tensor $C$ via singular value decomposition as
\begin{equation}
  C_{q_1,q_2}^{s^{[l]}_{\rm out}, s^{[l+1]}_{\rm out} ,s^{[l]}_{\rm in}, s^{[l+1]}_{\rm in}}
  = V_{q_1,\mu}^{s^{[l]}_{\rm out},s^{[l]}_{\rm in}} \,S_\mu \,\bar{W}_{q_2,\mu}^{s^{[l+1]}_{\rm out},s^{[l+1]}_{\rm in}},
\end{equation}
and we calculate the Shannon entropy of the induced discrete probability measure $p_\mu = S_{\mu}^2 / \sum_{\mu'} S_{\mu'}^2$.
Ultimately, we use such entropy $\mathcal{E}$ as a figure of merit for the direct search, and try to
find the $U$ which minimizes that quantity.
To do this, we apply a Nelder-Mead simplex algorithm, over the space $U(k)$ of $k \times k$ unitary matrices.

\subsection{Error analysis}\label{supp:trac_cert}
Our goal in this section is to derive an error bound for the algorithm presented in the main body of the paper, that includes the contribution from the Trotter-Suzuki approximation as well as the contribution from the local compression steps. Our figure of merit in this context is the fidelity as well as the trace-norm distance of the evolved state given by the algorithm in comparison to the state after the full time evolution.
Density operators are positive by definition, so we know that $\rho$ can be decomposed as $\rho = X^\dagger X$. Hence, our decomposition of $\rho$ into local tensor $A^{[l]}$ can be seen as choosing a particular $X$. Note that bond dimensions of exact local purifications can grow with the system size for certain density operators \cite{Gemma}. However, we assume that the initial state of our simulation already has a locally purified form, which is fulfilled, e.g., for product states and maximally mixed states.

From Lemma \ref{lem:2normdist} we see that in order to obtain a one-norm bound for $\rho$ we have to control the $2$-norm distance between two tensors $X$ and $Y$. To this end let us first define the compression error that introduced by throwing away singular values.

\begin{definition}[Discarded weight]\label{def:dw}
For a local tensor $\locT$ that is compressed by discarding singular values with respect to one of its indices, we define the \emph{discarded weight} $\delta$ of the compression as the introduced two norm error
  \begin{align}
   \dw \coloneqq \left({\sum_{s_i \text{ discarded}} s_i^2}\right)^{1/2}\;.
  \end{align}
\end{definition}

Equivalently, the discarded weight is the Frobenius norm error introduced by discarding the singular values. The following lemma also takes the renormalisation included in a compression step into account.

\begin{lemma}[Discarded weight]\label{lem:discarded_weight}
  Let $X$ be a locally purified description of a quantum state $\rho$ with local tensors $\{A^{[l]}\}$ that is in mixed canonical form with respect to a local tensor $A^{[l_{cp}]}$. Then, if $\widetilde X$  denotes the tensor, where $A^{[l_{cp}]}$ is compressed by a discarded weight $\dw$ either with respect to its Kraus or bond dimensions, we have
  \begin{equation}
    \normb{X-\tilde X}^2_2 = 2 \left(1-\sqrt{1-\dw^2}\right)  .
    \label{eq:estimator}
  \end{equation}
  \end{lemma}
\begin{proof}
\begin{figure}
  \centering
  \includegraphics[width=0.7\linewidth,angle=270]{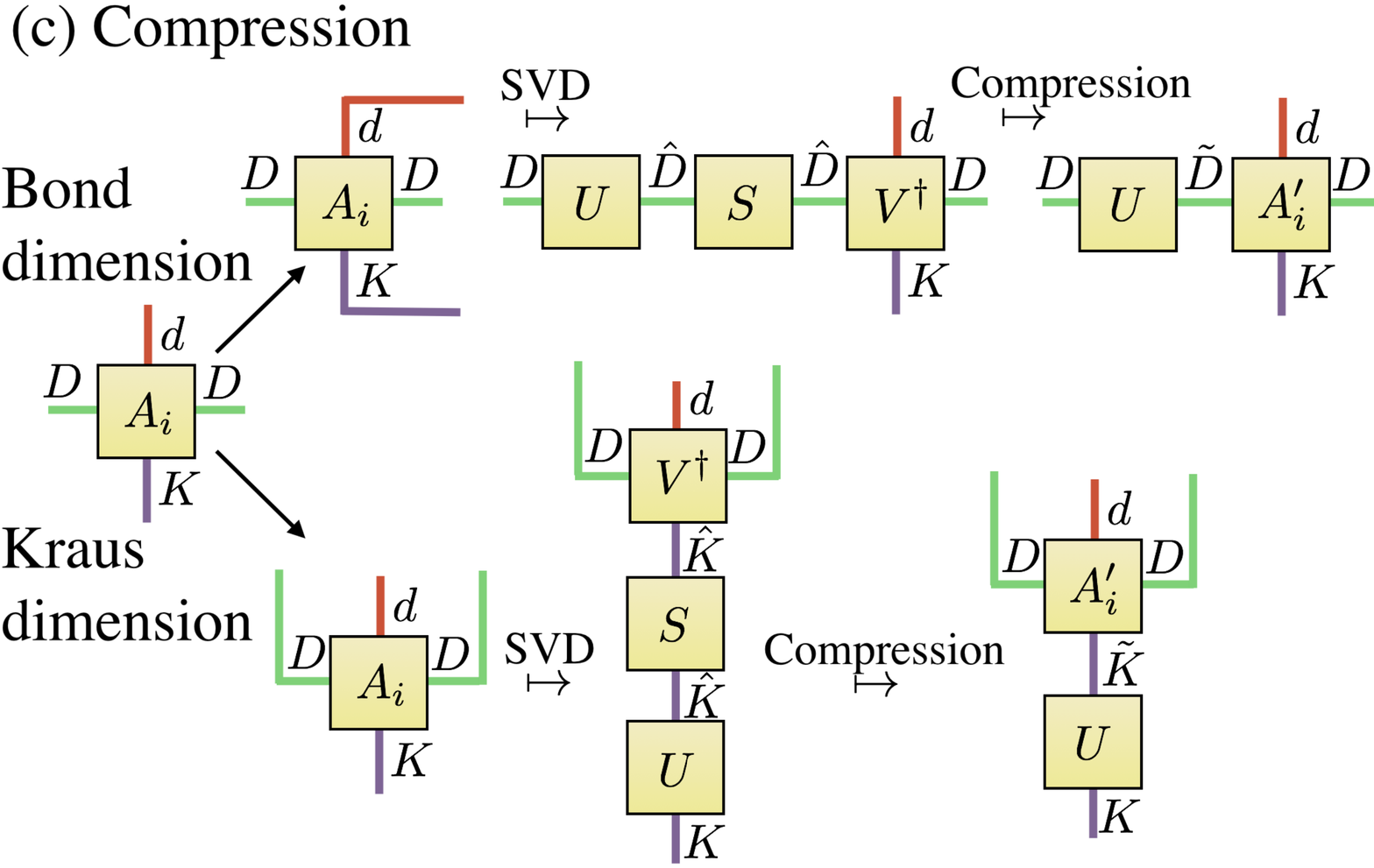}
  \caption{Compression schemes for Kraus and bond dimension. Note that for the compression of the Kraus-dimension we can drop the remaining unitary $U$ due to gauge freedom.}
  \label{fig:app:comp}
\end{figure}
  The compression of $X$ with respect to the local tensor $A^{[l_{cp}]}$ is done via a singular value decomposition $V S U^\dagger$ of $A^{[l_{cp}]}$ with respect to a specific distribution of its tensor indices depending on whether we compress the bond or the Kraus dimension (see Fig.\ \ref{fig:app:comp}). The diagonal matrix $S$ contains the singular values of $A^{[l_{cp}]}$ in decreasing order.

  The compression step itself consists of discarding the $L$ smallest singular values $s_l$ of $S$  such that $\dw =
  ({\sum_{l=n-L+1}^n s_l^2})^{1/2}$. This can be done by substituting the diagonal matrix $S$ by a matrix $\tilde S$ where the last $L$ diagonal entries have been set to zero. In order to renormalize the corresponding full tensor $\widetilde X$ we  use the assumption that $X$ was in a mixed canonical form with respect to $A^{[l_{cp}]}$, which implies that
  \begin{align}
  \norm{X} = \Tr(X^\dagger X)=\Tr(S^2)\;.
\end{align}
 Since $\Tr(\widetilde S^2) + \dw^2 = 1$ we have to divide $\widetilde S$ by $\sqrt{1-\dw^2}$ in order to ensure $\normn{\widetilde{X}}=1$.
  Since $X$ and $\widetilde X$ are now normalized we just have to compute their Hilbert-Schmidt-overlap in order to determine their $2$-norm-distance. However, since all local tensors $A^{[l]}$ for $l\neq l_{cp}$ are identical for $X$ and $\widetilde X$ and both are in mixed canonical form, their scalar product can be computed directly from $S$ and $\widetilde S$ via (see also Fig.\ \ref{fig:app:HSscp})
  \begin{eqnarray}
    \Scp{X}{Y} = \Tr(S \tilde S) = \frac{\sum_{l=1}^{n-L} s_l^2}{\sqrt{1-\dw^2}} = \sqrt{1-\dw^2}\;.
  \end{eqnarray}
  Using again Eq.\ \eqref{eq:2nscp} finishes the proof.
  \begin{figure}
  \centering
     \begin{overpic}[width=0.8\linewidth, unit=1pt]{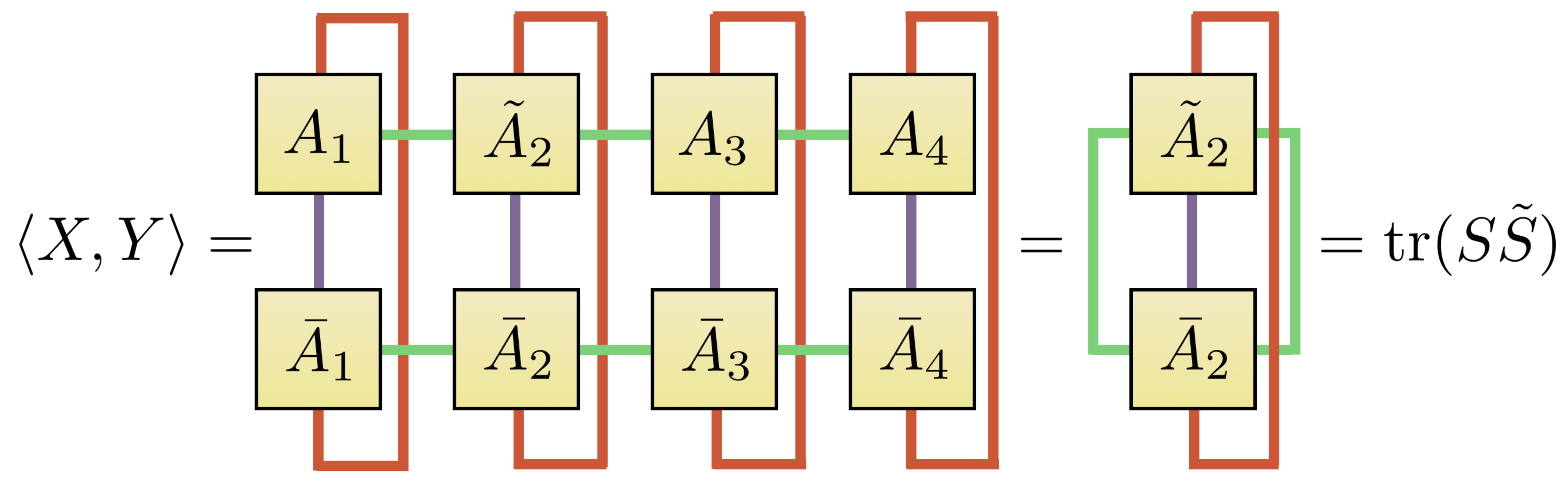}
     \put(65,14){\small $=$}
     \end{overpic}
    \caption{ Hilbert Schmidt inner product of a locally purified tensor $X$ with its compression $\widetilde X$.}
    \label{fig:app:HSscp}
  \end{figure}
\end{proof}

Equipped with these prerequisites we can finally write down the error bounds for our algorithm in terms of the trace norm.

\begin{theorem}[Trace norm certificate]\label{thm:certificate}
 Let
 \begin{equation}
  \L = \sum_{l=1}^{N-1} \L^{[l,l+1]}
 \end{equation}
 be a nearest neighbor Liouvillian on a chain of $N$ spins, where the local terms are bounded as
 $\dnorm{\L^{l,l+1}} \leq b$. Moreover, let $\rho=X X^\dagger$ be some initial state with purification operator $X$. Denote by $\rho_t \coloneqq \e^{t\L}(\rho)$ the exact time evolved state and by $\tilde \rho_t$ the state evolved according to our algorithm with $m$ times steps in the second order Trotter-Suzuki approximation and assume that all discarded weights are bounded by $\dw_{\max{}}$. Then
 \begin{equation}
    \normn{\rho_t - \tilde \rho_t}_1
    \leq
    \frac{(tb)^3 N^2}{4m^2} + 6 (2m+1) N \dw \, .
 \end{equation}
 \end{theorem}

Before we come to the proof, we briefly discuss the theorem. It provides an upper bound on the error of our simulation scaling polynomially in all parameters. As our simulation is done in purification space, we are able to measure the error in trace norm, which is the operationally relevant norm as it provides a uniform error bound for all expectation values. The most crucial parameter in our error bound is the discarded weight $\dw$. Its value is determined during the runtime of the simulation. Providing good apriori estimates for $\dw$ would show that Markovian dynamics can be simulated efficiently in the system size and remains an open research question. This would, e.g., also show that stationary states of rapidly mixing local Liouvillians \cite{Mixing} can be computed efficiently by our algorithm.

Of course, the error bound of the theorem is a worst-case estimate. The bound \eqref{eq:commutator_bound} and the uniform bound on the discarded weights are usually not tight and can be evaluated explicitly, e.g., during the runtime of the algorithm.

In addition, inequality \eqref{eq:fidelity_trace_norm} can be used to obtain the identical error bound in terms of the fidelity.

\begin{proof}[Proof of theorem \ref{thm:certificate}]
Let $\rho$ be the initial state, $t=\Delta t/m$ the simulation time and consider the local Liouvillian \eqref{eq:local_Liouvillian} generating the time evolution with nearest neighbor Liouvillians $\L^{[l,l+1]}$.
The goal is to bound the trace-norm difference between the exact time evolved state
\begin{equation}
 \rho_t \coloneqq \e^{t \L}(\rho)
\end{equation}
and the output state of our simulation scheme $\tilde \rho_t$. This is done in two steps: first we approximate $\e^{t \L}$ with its Trotter-Suzuki approximation $(T^\Trot_{\Delta t})^m$ and in the second step we approximate $(T^\Trot_{\Delta t})^m (\rho)$ by the output $\tilde \rho_t$ of our algorithm.

The Trotter-Suzuki approximation is, as in Eq.~\eqref{eq:Trotter_even_odd},
\begin{equation}
 T^\Trot_{\Delta t}
 \coloneqq \e^{\Delta t\, \mathcal{L}_o /2} \e^{\Delta t\, \mathcal{L}_e} \e^{\Delta t\, \mathcal{L}_o /2} ,
\end{equation}
where $\L_e$ and $\L_o$ are the even and odd Liouvillian terms \eqref{eq:L_even_odd}. Lemma~\ref{lem:second_order} implies that
\begin{align}
 \dnorm{\e^{\Delta t\, \L} - T^\Trot_{\Delta t} }
 &\leq
 \frac{\Delta t^3}{3} \dnorm{[\L_o,\L_e]} \, \left( \dnorm{\L_e} + \kw 2 \dnorm{\L_o} \right)\, .
\end{align}
Moreover,
\begin{align}
 \dnorm{\L_e} + \kw 2 \dnorm{\L_o}
 &\leq
 \sum_l \left(\dnormb{\L^{[2l,2l+1]}} + \kw2 \dnormb{\L^{[2l-1,2l]}}\right)\nonumber
 \\
 &\leq
\frac{3 N}4 b
\end{align}
and
\begin{align}
 \dnorm{[\L_o,\L_e]}
 &\leq
%  \sum_l \dnormb{\bigl[\L^{[2l,2l+1]}, \L^{2l-1,2l}+\L^{2l+1,2l+2} \bigr]}
%  \\
%  &=
 \sum_{l=1}^{N-2} \dnormb{\bigl[\L^{[l,l+1]}, \L^{[l+1,l+2]} \bigr]}\nonumber
 \\
 &\leq 2(N-2)b^2 \, .\label{eq:commutator_bound}
\end{align}
Together, this yields
\begin{equation}
 \dnorm{\e^{\Delta t\, \L} - T^\Trot_{\Delta t} }
 \leq
 \frac{\Delta t^3}{2} N^2 b^3 \, .
\end{equation}
Note that for any sub-multiplicative norm $\norm{\argdot}$
\begin{align}
\normn{AB-\tilde A \tilde B}\nonumber
&\leq
\normn{AB-A \tilde B} + \normn{A\tilde B -\tilde A \tilde B}
\\&
\leq
\normn{A}\normn{B-\tilde B} + \normn{A-\tilde A}\norm{B} \, .\label{eq:triangle_inequality}
\end{align}
Next, we use the fact that quantum channels are normalized in diamond norm and iterate Eq.~\eqref{eq:triangle_inequality} to obtain
\begin{align}
  \dnorm{\e^{t \L} - (T^\Trot_{\Delta t})^m }
  &\leq
  \frac{m\,\Delta t^3}{4}  N^2 b^3\nonumber
  \\
  &=
  \frac{(tb)^3 N^2}{4m^2} \, . \label{eq:TZ-error_in_pf}
\end{align}
Now we come to the second step and bound the error due to the approximation of
$(T^\Trot_{\Delta t})^m (\rho)$ with the state $\tilde \rho_{t}$ obtained after the compression steps of our algorithm. Note that
\begin{equation}\label{eq:TZ-circuit}
(T^\Trot_{\Delta t})^m
 =
 \e^{\Delta t\, \L_o /2}
 \left(\e^{\Delta t\, \L_e} \e^{\Delta t\, \L_o} \right)^{m-1}
 \e^{\Delta t\, \L_e} \e^{\Delta t\, \L_o /2}  \, ,
\end{equation}
i.e., $2m+1$ \emph{layers} $\e^{\Delta t \, \K^{(k)}}$ with $\K^{(k)} \in \{\L_e, \L_o, \L_o/2 \}$ and $k \in [2m+1]$ need to be applied to $\rho$.
In the algorithm, the initial state $\rho = XX^\dagger$ needs to be given in terms of the purification operator $X$ from Eq.\ \eqref{eq:X_in_terms_of_A}. With $\tilde X^{(0)} \coloneqq X^{(0)} \coloneqq X$ we denote the initial purification operator, by $X^{(k)}$ the updated purification operators, and by $\tilde X^{(k)}$ the compressed updated purification operators. The corresponding local variational tensors are denoted by $A^{(k)[l]}$ and $\tilde A^{(k)[l]}$, respectively (as in Eq.\ \eqref{eq:X_in_terms_of_A}). In the algorithm, the following is done for layers $k=1, 2, \dots, 2m+1$: In the $k$-th step, first $\tilde X^{(k-1)}$ is updated by the $k$-th layer to obtain $X^{(k)}$ and, second, for $l=1,2,\dots, N$ and $j=1,2,3$ the $j$-th virtual index of the local tensor $A^{(k)[l]}$ of $X^{(k)}$ is compressed to $\tilde A^{(k)[l]}$ by discarding a weight $\delta_{k,l,j}$. Due to the triangle inequality the two norm errors introduced by the discarded weights can at most sum up. With Lemma~\ref{lem:discarded_weight} this yields,
\begin{equation}\label{eq:diff_Xk}
 \normb{X^{(k)} - \tilde X^{(k)}}_2 \leq  \sum_{l=1}^N \sum_{j=1}^3
 \left({2(1-\sqrt{1-\dw_{k,l,j}^2})}\right)^{1/2}\, .
\end{equation}
Now let $\rho^{(k)}$ be the state obtain by applying the first $k$ layers successively to $\rho$, so that $\rho^{(2m+1)} = (T^\Trot_{\Delta t})^m(\rho)$, see Eq.~\eqref{eq:TZ-circuit}. Then, using that $\normb{\e^{\Delta t \K^{(k)} }(Y)}_1 \leq \norm{Y}_1$ for any operator $Y$,
the triangle inequality, and Lemma~\ref{lem:2normdist}, we obtain that
\begin{widetext}
\begin{align}
 \normb{\rho^{(k)} - X^{(k)}X^{(k)\dagger}}_1
 &=
 \norm{\e^{\Delta t \K^{(k)} } \Bigl(\rho^{(k-1))} - \tilde X^{(k-1)} \tilde X^{(k-1)\dagger}\Bigr)}_1
 \\ \nonumber
 &\leq
 \norm{\rho^{(k-1))} -  \tilde X^{(k-1)} \tilde X^{(k-1)\dagger} }_1
 \\ \nonumber
 &\leq
 \norm{\rho^{(k-1))} - X^{(k-1)} X^{(k-1)\dagger} }_1
  +\norm{X^{(k-1)} X^{(k-1)\dagger}-\tilde X^{(k-1)} \tilde X^{(k-1)\dagger}}_1
  \\ \nonumber
 &\leq
 \norm{\rho^{(k-1))} - X^{(k-1)} X^{(k-1)\dagger} }_1
  + \sqrt{2} \norm{ X^{(k-1)} - \tilde X^{(k-1)}}_2 \, .
\end{align}
Iterating this inequality and using Eq.~\eqref{eq:diff_Xk} yields
\begin{align}
 \norm{(T^\Trot_{\Delta t})^m(\rho) - X^{(2m+1)}X^{(2m+1)\dagger}}_1
 \leq
 2\sum_{k=1}^{2m+1}\sum_{l=1}^N \sum_{j=1}^3 \left({(1-\sqrt{1-\dw_{k,l,j}^2})}\right)^{1/2} , .
\end{align}
We use the bound $\dw_{k,l,j} \leq \dw_{\max{}}$, the triangle inequality, Eq.~\eqref{eq:TZ-error_in_pf}, and that
$({(1-\sqrt{1-\dw^2})})^{1/2} \leq \dw$ for $\dw \in [0,1]$ to obtain
\begin{align}
 \norm{\e^{t\L}(\rho) - X^{(2m+1)}X^{(2m+1) \dagger}}_1
 &\leq
 \frac{(tb)^3 N^2}{4m^2} + 6 (2m+1) N \left({(1-\sqrt{1-\dw^2})}\right)^{1/2}\nonumber
 \\
 &\leq
\frac{(tb)^3 N^2}{4m^2} + 6 (2m+1) N \dw \, .
\end{align}
\end{widetext}
\end{proof}

\begin{figure}[ht]
 \begin{center}
\includegraphics[height=5.5cm,angle=-90]{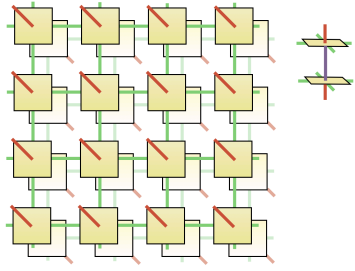}
 \end{center}
\caption{ \label{fig:2d}
The tensor networks of positive PEPO (projected entangled pair operators) to capture two-dimensional dissipative
quantum systems. The small image on the right hand side represents a single local purification, depicting
both the edges taking the bond and Kraus dimensions.
}
\end{figure}
\subsection{Prospects for the simulation of two-dimensional open systems}\label{sm:2d}
Our ansatz class of locally purified states generalizes to two spatial dimensions in a similar way as MPS are generalized by PEPS by introducing additional bond indices for the additional neighbors, see Fig.\ \ref{fig:2d}.
The second order dissipative Trotter-Suzuki approximation and the update steps based on the local Kraus operators of the local channels also directly generalize to higher dimensions. For the most natural case of a nearest-neighbor model with two body interactions in two spatial dimensions, this procedure would require four Trotter-Suzuki layers to implement.

For the compression steps, the canonical representation is needed. As the 2D version of our tensor network states ansatz class is a direct generalization of PEPS, it also suffers from the same numerical issues, such as exact contractions being $\#\P$-hard in the worst case \cite{SchWolVer07}. However, in a similar way as promising numerical simulations are done with PEPS, we expect the 2D generalization of our simulations to lead to promising simulations. The generalization of such approximate sweeping methods seems straight-forward for the bond indices of the local tensors. However, the translation to Kraus-dimensions seems less obvious, but might still be possible via singular value decomposition.

%%% =============================================================================================


\begin{thebibliography}{57}%
\makeatletter
\providecommand \@ifxundefined [1]{%
 \@ifx{#1\undefined}
}%
\providecommand \@ifnum [1]{%
 \ifnum #1\expandafter \@firstoftwo
 \else \expandafter \@secondoftwo
 \fi
}%
\providecommand \@ifx [1]{%
 \ifx #1\expandafter \@firstoftwo
 \else \expandafter \@secondoftwo
 \fi
}%
\providecommand \natexlab [1]{#1}%
\providecommand \enquote  [1]{``#1''}%
\providecommand \bibnamefont  [1]{#1}%
\providecommand \bibfnamefont [1]{#1}%
\providecommand \citenamefont [1]{#1}%
\providecommand \href@noop [0]{\@secondoftwo}%
\providecommand \href [0]{\begingroup \@sanitize@url \@href}%
\providecommand \@href[1]{\@@startlink{#1}\@@href}%
\providecommand \@@href[1]{\endgroup#1\@@endlink}%
\providecommand \@sanitize@url [0]{\catcode `\\12\catcode `\$12\catcode
  `\&12\catcode `\#12\catcode `\^12\catcode `\_12\catcode `\%12\relax}%
\providecommand \@@startlink[1]{}%
\providecommand \@@endlink[0]{}%
\providecommand \url  [0]{\begingroup\@sanitize@url \@url }%
\providecommand \@url [1]{\endgroup\@href {#1}{\urlprefix }}%
\providecommand \urlprefix  [0]{URL }%
\providecommand \Eprint [0]{\href }%
\providecommand \doibase [0]{http://dx.doi.org/}%
\providecommand \selectlanguage [0]{\@gobble}%
\providecommand \bibinfo  [0]{\@secondoftwo}%
\providecommand \bibfield  [0]{\@secondoftwo}%
\providecommand \translation [1]{[#1]}%
\providecommand \BibitemOpen [0]{}%
\providecommand \bibitemStop [0]{}%
\providecommand \bibitemNoStop [0]{.\EOS\space}%
\providecommand \EOS [0]{\spacefactor3000\relax}%
\providecommand \BibitemShut  [1]{\csname bibitem#1\endcsname}%
\let\auto@bib@innerbib\@empty
%</preamble>
\bibitem [{\citenamefont {Diehl}\ \emph {et~al.}(2008)\citenamefont {Diehl},
  \citenamefont {Micheli}, \citenamefont {Kantian}, \citenamefont {Kraus},
  \citenamefont {Buechler},\ and\ \citenamefont {Zoller}}]{DissipationZoller}%
  \BibitemOpen
  \bibfield  {author} {\bibinfo {author} {\bibfnamefont {S.}~\bibnamefont
  {Diehl}}, \bibinfo {author} {\bibfnamefont {A.}~\bibnamefont {Micheli}},
  \bibinfo {author} {\bibfnamefont {A.}~\bibnamefont {Kantian}}, \bibinfo
  {author} {\bibfnamefont {B.}~\bibnamefont {Kraus}}, \bibinfo {author}
  {\bibfnamefont {H.~P.}\ \bibnamefont {Buechler}}, \ and\ \bibinfo {author}
  {\bibfnamefont {P.}~\bibnamefont {Zoller}},\ }\href@noop {} {\bibfield
  {journal} {\bibinfo  {journal} {Nature Phys.}\ }\textbf {\bibinfo {volume}
  {4}},\ \bibinfo {pages} {878} (\bibinfo {year} {2008})}\BibitemShut {NoStop}%
\bibitem [{\citenamefont {Verstraete}\ \emph {et~al.}(2009)\citenamefont
  {Verstraete}, \citenamefont {Wolf},\ and\ \citenamefont
  {Cirac}}]{DissipationComputing}%
  \BibitemOpen
  \bibfield  {author} {\bibinfo {author} {\bibfnamefont {F.}~\bibnamefont
  {Verstraete}}, \bibinfo {author} {\bibfnamefont {M.~M.}\ \bibnamefont
  {Wolf}}, \ and\ \bibinfo {author} {\bibfnamefont {J.~I.}\ \bibnamefont
  {Cirac}},\ }\href@noop {} {\bibfield  {journal} {\bibinfo  {journal} {Nature
  Phys.}\ }\textbf {\bibinfo {volume} {5}},\ \bibinfo {pages} {633} (\bibinfo
  {year} {2009})}\BibitemShut {NoStop}%
\bibitem [{\citenamefont {Krauter}\ \emph {et~al.}(2011)\citenamefont
  {Krauter}, \citenamefont {Muschik}, \citenamefont {Jensen}, \citenamefont
  {Wasilewski}, \citenamefont {Petersen}, \citenamefont {Cirac},\ and\
  \citenamefont {Polzik}}]{DissipationEntanglement}%
  \BibitemOpen
  \bibfield  {author} {\bibinfo {author} {\bibfnamefont {H.}~\bibnamefont
  {Krauter}}, \bibinfo {author} {\bibfnamefont {C.~A.}\ \bibnamefont
  {Muschik}}, \bibinfo {author} {\bibfnamefont {K.}~\bibnamefont {Jensen}},
  \bibinfo {author} {\bibfnamefont {W.}~\bibnamefont {Wasilewski}}, \bibinfo
  {author} {\bibfnamefont {J.~M.}\ \bibnamefont {Petersen}}, \bibinfo {author}
  {\bibfnamefont {J.~I.}\ \bibnamefont {Cirac}}, \ and\ \bibinfo {author}
  {\bibfnamefont {E.~S.}\ \bibnamefont {Polzik}},\ }\href@noop {} {\bibfield
  {journal} {\bibinfo  {journal} {Phys. Rev. Lett.}\ }\textbf {\bibinfo
  {volume} {107}},\ \bibinfo {pages} {080503} (\bibinfo {year}
  {2011})}\BibitemShut {NoStop}%
\bibitem [{\citenamefont {Plenio}\ and\ \citenamefont
  {Huelga}(2002)}]{DissipationOld}%
  \BibitemOpen
  \bibfield  {author} {\bibinfo {author} {\bibfnamefont {M.~B.}\ \bibnamefont
  {Plenio}}\ and\ \bibinfo {author} {\bibfnamefont {S.~F.}\ \bibnamefont
  {Huelga}},\ }\href@noop {} {\bibfield  {journal} {\bibinfo  {journal} {Phys.
  Rev. Lett.}\ }\textbf {\bibinfo {volume} {88}},\ \bibinfo {pages} {197901}
  (\bibinfo {year} {2002})}\BibitemShut {NoStop}%
\bibitem [{\citenamefont {Eisert}\ and\ \citenamefont
  {Prosen}(2010)}]{DissipationCritical}%
  \BibitemOpen
  \bibfield  {author} {\bibinfo {author} {\bibfnamefont {J.}~\bibnamefont
  {Eisert}}\ and\ \bibinfo {author} {\bibfnamefont {T.}~\bibnamefont
  {Prosen}},\ }\href@noop {} {} (\bibinfo {year} {2010}),\ \bibinfo {note}
  {arXiv:1012.5013}\BibitemShut {NoStop}%
\bibitem [{\citenamefont {Kastoryano}\ \emph {et~al.}(2013)\citenamefont
  {Kastoryano}, \citenamefont {Wolf},\ and\ \citenamefont {Eisert}}]{Timing}%
  \BibitemOpen
  \bibfield  {author} {\bibinfo {author} {\bibfnamefont {M.~J.}\ \bibnamefont
  {Kastoryano}}, \bibinfo {author} {\bibfnamefont {M.~M.}\ \bibnamefont
  {Wolf}}, \ and\ \bibinfo {author} {\bibfnamefont {J.}~\bibnamefont
  {Eisert}},\ }\href@noop {} {\bibfield  {journal} {\bibinfo  {journal} {Phys.
  Rev. Lett.}\ }\textbf {\bibinfo {volume} {103}},\ \bibinfo {pages} {110501}
  (\bibinfo {year} {2013})}\BibitemShut {NoStop}%
\bibitem [{\citenamefont {Diehl}\ \emph {et~al.}(2011)\citenamefont {Diehl},
  \citenamefont {Rico}, \citenamefont {Baranov},\ and\ \citenamefont
  {Zoller}}]{DissipationTopologyZoller}%
  \BibitemOpen
  \bibfield  {author} {\bibinfo {author} {\bibfnamefont {S.}~\bibnamefont
  {Diehl}}, \bibinfo {author} {\bibfnamefont {E.}~\bibnamefont {Rico}},
  \bibinfo {author} {\bibfnamefont {M.~A.}\ \bibnamefont {Baranov}}, \ and\
  \bibinfo {author} {\bibfnamefont {P.}~\bibnamefont {Zoller}},\ }\href@noop {}
  {\bibfield  {journal} {\bibinfo  {journal} {Nature Phys.}\ }\textbf {\bibinfo
  {volume} {7}},\ \bibinfo {pages} {971} (\bibinfo {year} {2011})}\BibitemShut
  {NoStop}%
\bibitem [{\citenamefont {Bardyn}\ \emph {et~al.}(2013)\citenamefont {Bardyn},
  \citenamefont {Baranov}, \citenamefont {Kraus}, \citenamefont {Rico},
  \citenamefont {Imamoglu}, \citenamefont {Zoller},\ and\ \citenamefont
  {Diehl}}]{TopologyDissipationLong}%
  \BibitemOpen
  \bibfield  {author} {\bibinfo {author} {\bibfnamefont {C.-E.}\ \bibnamefont
  {Bardyn}}, \bibinfo {author} {\bibfnamefont {M.~A.}\ \bibnamefont {Baranov}},
  \bibinfo {author} {\bibfnamefont {C.~V.}\ \bibnamefont {Kraus}}, \bibinfo
  {author} {\bibfnamefont {E.}~\bibnamefont {Rico}}, \bibinfo {author}
  {\bibfnamefont {A.}~\bibnamefont {Imamoglu}}, \bibinfo {author}
  {\bibfnamefont {P.}~\bibnamefont {Zoller}}, \ and\ \bibinfo {author}
  {\bibfnamefont {S.}~\bibnamefont {Diehl}},\ }\href@noop {} {\bibfield
  {journal} {\bibinfo  {journal} {New J. Phys.}\ }\textbf {\bibinfo {volume}
  {15}},\ \bibinfo {pages} {085001} (\bibinfo {year} {2013})}\BibitemShut
  {NoStop}%
\bibitem [{\citenamefont {Prosen}(2011)}]{ProsenDissipativeStationaryXXZ}%
  \BibitemOpen
  \bibfield  {author} {\bibinfo {author} {\bibfnamefont {T.}~\bibnamefont
  {Prosen}},\ }\href@noop {} {\bibfield  {journal} {\bibinfo  {journal} {Phys.
  Rev. Lett.}\ }\textbf {\bibinfo {volume} {107}},\ \bibinfo {pages} {137201}
  (\bibinfo {year} {2011})}\BibitemShut {NoStop}%
\bibitem [{\citenamefont {Fannes}\ \emph {et~al.}(1992)\citenamefont {Fannes},
  \citenamefont {Nachtergaele},\ and\ \citenamefont {Werner}}]{FCS}%
  \BibitemOpen
  \bibfield  {author} {\bibinfo {author} {\bibfnamefont {M.}~\bibnamefont
  {Fannes}}, \bibinfo {author} {\bibfnamefont {B.}~\bibnamefont
  {Nachtergaele}}, \ and\ \bibinfo {author} {\bibfnamefont {R.~F.}\
  \bibnamefont {Werner}},\ }\href@noop {} {\bibfield  {journal} {\bibinfo
  {journal} {Lett. Math. Phys.}\ }\textbf {\bibinfo {volume} {25}},\ \bibinfo
  {pages} {249} (\bibinfo {year} {1992})}\BibitemShut {NoStop}%
\bibitem [{\citenamefont {Perez-Garcia}\ \emph {et~al.}(2007)\citenamefont
  {Perez-Garcia}, \citenamefont {Verstraete}, \citenamefont {Wolf},\ and\
  \citenamefont {Cirac}}]{MPSReps}%
  \BibitemOpen
  \bibfield  {author} {\bibinfo {author} {\bibfnamefont {D.}~\bibnamefont
  {Perez-Garcia}}, \bibinfo {author} {\bibfnamefont {F.}~\bibnamefont
  {Verstraete}}, \bibinfo {author} {\bibfnamefont {M.~M.}\ \bibnamefont
  {Wolf}}, \ and\ \bibinfo {author} {\bibfnamefont {J.~I.}\ \bibnamefont
  {Cirac}},\ }\href@noop {} {\bibfield  {journal} {\bibinfo  {journal} {Quantum
  Inf. Comput.}\ }\textbf {\bibinfo {volume} {7}},\ \bibinfo {pages} {401}
  (\bibinfo {year} {2007})}\BibitemShut {NoStop}%
\bibitem [{\citenamefont {Murg}\ \emph {et~al.}(2007)\citenamefont {Murg},
  \citenamefont {Verstraete},\ and\ \citenamefont {Cirac}}]{Murg}%
  \BibitemOpen
  \bibfield  {author} {\bibinfo {author} {\bibfnamefont {V.}~\bibnamefont
  {Murg}}, \bibinfo {author} {\bibfnamefont {F.}~\bibnamefont {Verstraete}}, \
  and\ \bibinfo {author} {\bibfnamefont {J.~I.}\ \bibnamefont {Cirac}},\
  }\href@noop {} {\bibfield  {journal} {\bibinfo  {journal} {Phys. Rev. A}\
  }\textbf {\bibinfo {volume} {75}},\ \bibinfo {pages} {033605} (\bibinfo
  {year} {2007})}\BibitemShut {NoStop}%
\bibitem [{\citenamefont {Schollw\"ock}(2011)}]{MPSRev}%
  \BibitemOpen
  \bibfield  {author} {\bibinfo {author} {\bibfnamefont {U.}~\bibnamefont
  {Schollw\"ock}},\ }\href@noop {} {\bibfield  {journal} {\bibinfo  {journal}
  {Ann. Phys.}\ }\textbf {\bibinfo {volume} {326}},\ \bibinfo {pages} {96}
  (\bibinfo {year} {2011})}\BibitemShut {NoStop}%
\bibitem [{\citenamefont {Dawson}\ \emph {et~al.}(2008)\citenamefont {Dawson},
  \citenamefont {Eisert},\ and\ \citenamefont {Osborne}}]{Flow}%
  \BibitemOpen
  \bibfield  {author} {\bibinfo {author} {\bibfnamefont {C.~M.}\ \bibnamefont
  {Dawson}}, \bibinfo {author} {\bibfnamefont {J.}~\bibnamefont {Eisert}}, \
  and\ \bibinfo {author} {\bibfnamefont {T.~J.}\ \bibnamefont {Osborne}},\
  }\href@noop {} {\bibfield  {journal} {\bibinfo  {journal} {Phys. Rev. Lett.}\
  }\textbf {\bibinfo {volume} {100}},\ \bibinfo {pages} {130501} (\bibinfo
  {year} {2008})}\BibitemShut {NoStop}%
\bibitem [{\citenamefont {Eisert}\ \emph {et~al.}(2010)\citenamefont {Eisert},
  \citenamefont {Cramer},\ and\ \citenamefont {Plenio}}]{AreaReview}%
  \BibitemOpen
  \bibfield  {author} {\bibinfo {author} {\bibfnamefont {J.}~\bibnamefont
  {Eisert}}, \bibinfo {author} {\bibfnamefont {M.}~\bibnamefont {Cramer}}, \
  and\ \bibinfo {author} {\bibfnamefont {M.~B.}\ \bibnamefont {Plenio}},\
  }\href@noop {} {\bibfield  {journal} {\bibinfo  {journal} {Rev. Mod. Phys.}\
  }\textbf {\bibinfo {volume} {82}},\ \bibinfo {pages} {277} (\bibinfo {year}
  {2010})}\BibitemShut {NoStop}%
\bibitem [{\citenamefont {White}(1992)}]{DMRGWhite92}%
  \BibitemOpen
  \bibfield  {author} {\bibinfo {author} {\bibfnamefont {S.~R.}\ \bibnamefont
  {White}},\ }\href@noop {} {\bibfield  {journal} {\bibinfo  {journal} {Phys.
  Rev. Lett.}\ }\textbf {\bibinfo {volume} {69}},\ \bibinfo {pages} {2863}
  (\bibinfo {year} {1992})}\BibitemShut {NoStop}%
\bibitem [{\citenamefont {\"Ostlund}\ and\ \citenamefont
  {Rommer}(1995)}]{MPSZero}%
  \BibitemOpen
  \bibfield  {author} {\bibinfo {author} {\bibfnamefont {S.}~\bibnamefont
  {\"Ostlund}}\ and\ \bibinfo {author} {\bibfnamefont {S.}~\bibnamefont
  {Rommer}},\ }\href@noop {} {\bibfield  {journal} {\bibinfo  {journal} {Phys.
  Rev. Lett.}\ }\textbf {\bibinfo {volume} {75}},\ \bibinfo {pages} {3537}
  (\bibinfo {year} {1995})}\BibitemShut {NoStop}%
\bibitem [{\citenamefont {Haegeman}\ \emph {et~al.}(2011)\citenamefont
  {Haegeman}, \citenamefont {Cirac}, \citenamefont {Osborne}, \citenamefont
  {Pizorn}, \citenamefont {Verschelde},\ and\ \citenamefont
  {Verstraete}}]{Lagrangian}%
  \BibitemOpen
  \bibfield  {author} {\bibinfo {author} {\bibfnamefont {J.}~\bibnamefont
  {Haegeman}}, \bibinfo {author} {\bibfnamefont {J.~I.}\ \bibnamefont {Cirac}},
  \bibinfo {author} {\bibfnamefont {T.~J.}\ \bibnamefont {Osborne}}, \bibinfo
  {author} {\bibfnamefont {I.}~\bibnamefont {Pizorn}}, \bibinfo {author}
  {\bibfnamefont {H.}~\bibnamefont {Verschelde}}, \ and\ \bibinfo {author}
  {\bibfnamefont {F.}~\bibnamefont {Verstraete}},\ }\href@noop {} {\bibfield
  {journal} {\bibinfo  {journal} {Phys. Rev. Lett.}\ }\textbf {\bibinfo
  {volume} {107}},\ \bibinfo {pages} {070601} (\bibinfo {year}
  {2011})}\BibitemShut {NoStop}%
\bibitem [{\citenamefont {Orus}(2014)}]{orus2013}%
  \BibitemOpen
  \bibfield  {author} {\bibinfo {author} {\bibfnamefont {R.}~\bibnamefont
  {Orus}},\ }\href@noop {} {\bibfield  {journal} {\bibinfo  {journal} {Ann.
  Phys.}\ }\textbf {\bibinfo {volume} {349}},\ \bibinfo {pages} {117} (\bibinfo
  {year} {2014})}\BibitemShut {NoStop}%
\bibitem [{\citenamefont {Rizzi}\ \emph {et~al.}(2008)\citenamefont {Rizzi},
  \citenamefont {Montangero},\ and\ \citenamefont {Vidal}}]{Renegade}%
  \BibitemOpen
  \bibfield  {author} {\bibinfo {author} {\bibfnamefont {M.}~\bibnamefont
  {Rizzi}}, \bibinfo {author} {\bibfnamefont {S.}~\bibnamefont {Montangero}}, \
  and\ \bibinfo {author} {\bibfnamefont {G.}~\bibnamefont {Vidal}},\
  }\href@noop {} {\bibfield  {journal} {\bibinfo  {journal} {Phys. Rev. A}\
  }\textbf {\bibinfo {volume} {77}},\ \bibinfo {pages} {052328} (\bibinfo
  {year} {2008})}\BibitemShut {NoStop}%
\bibitem [{\citenamefont {Rizzi}\ \emph {et~al.}(2010)\citenamefont {Rizzi},
  \citenamefont {Montangero}, \citenamefont {Silvi}, \citenamefont
  {Giovannetti},\ and\ \citenamefont {Fazio}}]{RizMera}%
  \BibitemOpen
  \bibfield  {author} {\bibinfo {author} {\bibfnamefont {M.}~\bibnamefont
  {Rizzi}}, \bibinfo {author} {\bibfnamefont {S.}~\bibnamefont {Montangero}},
  \bibinfo {author} {\bibfnamefont {P.}~\bibnamefont {Silvi}}, \bibinfo
  {author} {\bibfnamefont {V.}~\bibnamefont {Giovannetti}}, \ and\ \bibinfo
  {author} {\bibfnamefont {R.}~\bibnamefont {Fazio}},\ }\href {\doibase
  10.1088/1367-2630/12/7/075018} {\bibfield  {journal} {\bibinfo  {journal}
  {New J. Phys.}\ }\textbf {\bibinfo {volume} {12}},\ \bibinfo {pages} {075018}
  (\bibinfo {year} {2010})}\BibitemShut {NoStop}%
\bibitem [{\citenamefont {Gerster}\ \emph {et~al.}(2014)\citenamefont
  {Gerster}, \citenamefont {Silvi}, \citenamefont {Rizzi}, \citenamefont
  {Fazio}, \citenamefont {Calarco},\ and\ \citenamefont
  {Montangero}}]{Adaptive}%
  \BibitemOpen
  \bibfield  {author} {\bibinfo {author} {\bibfnamefont {M.}~\bibnamefont
  {Gerster}}, \bibinfo {author} {\bibfnamefont {P.}~\bibnamefont {Silvi}},
  \bibinfo {author} {\bibfnamefont {M.}~\bibnamefont {Rizzi}}, \bibinfo
  {author} {\bibfnamefont {R.}~\bibnamefont {Fazio}}, \bibinfo {author}
  {\bibfnamefont {T.}~\bibnamefont {Calarco}}, \ and\ \bibinfo {author}
  {\bibfnamefont {S.}~\bibnamefont {Montangero}},\ }\href@noop {} {\bibfield
  {journal} {\bibinfo  {journal} {Phys. Rev. B}\ }\textbf {\bibinfo {volume}
  {90}},\ \bibinfo {pages} {125154} (\bibinfo {year} {2014})}\BibitemShut
  {NoStop}%
\bibitem [{\citenamefont {Corboz}\ \emph {et~al.}(2010)\citenamefont {Corboz},
  \citenamefont {Evenbly}, \citenamefont {Verstraete},\ and\ \citenamefont
  {Vidal}}]{MERAF1}%
  \BibitemOpen
  \bibfield  {author} {\bibinfo {author} {\bibfnamefont {P.}~\bibnamefont
  {Corboz}}, \bibinfo {author} {\bibfnamefont {G.}~\bibnamefont {Evenbly}},
  \bibinfo {author} {\bibfnamefont {F.}~\bibnamefont {Verstraete}}, \ and\
  \bibinfo {author} {\bibfnamefont {G.}~\bibnamefont {Vidal}},\ }\href@noop {}
  {\bibfield  {journal} {\bibinfo  {journal} {Phys. Rev. A (R)}\ }\textbf
  {\bibinfo {volume} {81}},\ \bibinfo {pages} {010303} (\bibinfo {year}
  {2010})}\BibitemShut {NoStop}%
\bibitem [{\citenamefont {Barthel}\ \emph {et~al.}(2009)\citenamefont
  {Barthel}, \citenamefont {Pineda},\ and\ \citenamefont {Eisert}}]{MERAF3}%
  \BibitemOpen
  \bibfield  {author} {\bibinfo {author} {\bibfnamefont {T.}~\bibnamefont
  {Barthel}}, \bibinfo {author} {\bibfnamefont {C.}~\bibnamefont {Pineda}}, \
  and\ \bibinfo {author} {\bibfnamefont {J.}~\bibnamefont {Eisert}},\
  }\href@noop {} {\bibfield  {journal} {\bibinfo  {journal} {Phys. Rev. A}\
  }\textbf {\bibinfo {volume} {80}},\ \bibinfo {pages} {042333} (\bibinfo
  {year} {2009})}\BibitemShut {NoStop}%
\bibitem [{\citenamefont {Barthel}\ \emph {et~al.}(2010)\citenamefont
  {Barthel}, \citenamefont {Kliesch},\ and\ \citenamefont
  {Eisert}}]{MERArealaw}%
  \BibitemOpen
  \bibfield  {author} {\bibinfo {author} {\bibfnamefont {T.}~\bibnamefont
  {Barthel}}, \bibinfo {author} {\bibfnamefont {M.}~\bibnamefont {Kliesch}}, \
  and\ \bibinfo {author} {\bibfnamefont {J.}~\bibnamefont {Eisert}},\
  }\href@noop {} {\bibfield  {journal} {\bibinfo  {journal} {Phys. Rev. Lett.}\
  }\textbf {\bibinfo {volume} {105}},\ \bibinfo {pages} {010502} (\bibinfo
  {year} {2010})}\BibitemShut {NoStop}%
\bibitem [{\citenamefont {Lubasch}\ \emph
  {et~al.}(2014{\natexlab{a}})\citenamefont {Lubasch}, \citenamefont {Cirac},\
  and\ \citenamefont {Ba\~nuls}}]{Lubasch2014}%
  \BibitemOpen
  \bibfield  {author} {\bibinfo {author} {\bibfnamefont {M.}~\bibnamefont
  {Lubasch}}, \bibinfo {author} {\bibfnamefont {J.~I.}\ \bibnamefont {Cirac}},
  \ and\ \bibinfo {author} {\bibfnamefont {M.-C.}\ \bibnamefont {Ba\~nuls}},\
  }\href {\doibase 10.1103/PhysRevB.90.064425} {\bibfield  {journal} {\bibinfo
  {journal} {Phys. Rev. B}\ }\textbf {\bibinfo {volume} {90}},\ \bibinfo
  {pages} {064425} (\bibinfo {year} {2014}{\natexlab{a}})}\BibitemShut
  {NoStop}%
\bibitem [{\citenamefont {Verstraete}\ \emph {et~al.}(2004)\citenamefont
  {Verstraete}, \citenamefont {Garcia-Ripoll},\ and\ \citenamefont
  {Cirac}}]{Mixed}%
  \BibitemOpen
  \bibfield  {author} {\bibinfo {author} {\bibfnamefont {F.}~\bibnamefont
  {Verstraete}}, \bibinfo {author} {\bibfnamefont {J.~J.}\ \bibnamefont
  {Garcia-Ripoll}}, \ and\ \bibinfo {author} {\bibfnamefont {J.~I.}\
  \bibnamefont {Cirac}},\ }\href@noop {} {\bibfield  {journal} {\bibinfo
  {journal} {Phys. Rev. Lett.}\ }\textbf {\bibinfo {volume} {93}},\ \bibinfo
  {pages} {207204} (\bibinfo {year} {2004})}\BibitemShut {NoStop}%
\bibitem [{\citenamefont {Zwolak}\ and\ \citenamefont
  {Vidal}(2004)}]{ZwolakVidal}%
  \BibitemOpen
  \bibfield  {author} {\bibinfo {author} {\bibfnamefont {M.}~\bibnamefont
  {Zwolak}}\ and\ \bibinfo {author} {\bibfnamefont {G.}~\bibnamefont {Vidal}},\
  }\href@noop {} {\bibfield  {journal} {\bibinfo  {journal} {Phys. Rev. Lett.}\
  }\textbf {\bibinfo {volume} {93}},\ \bibinfo {pages} {207205} (\bibinfo
  {year} {2004})}\BibitemShut {NoStop}%
\bibitem [{\citenamefont {Bonnes}\ \emph {et~al.}(2014)\citenamefont {Bonnes},
  \citenamefont {Charrier},\ and\ \citenamefont
  {L{\"a}uchli}}]{LaeuchliDissipative}%
  \BibitemOpen
  \bibfield  {author} {\bibinfo {author} {\bibfnamefont {L.}~\bibnamefont
  {Bonnes}}, \bibinfo {author} {\bibfnamefont {D.}~\bibnamefont {Charrier}}, \
  and\ \bibinfo {author} {\bibfnamefont {A.~M.}\ \bibnamefont {L{\"a}uchli}},\
  }\href@noop {} {\bibfield  {journal} {\bibinfo  {journal} {Phys. Rev. A}\
  }\textbf {\bibinfo {volume} {90}},\ \bibinfo {pages} {033612} (\bibinfo
  {year} {2014})}\BibitemShut {NoStop}%
\bibitem [{\citenamefont {Daley}(2014)}]{DaleyOpenReview}%
  \BibitemOpen
  \bibfield  {author} {\bibinfo {author} {\bibfnamefont {A.~J.}\ \bibnamefont
  {Daley}},\ }\href {\doibase 10.1080/00018732.2014.933502} {\bibfield
  {journal} {\bibinfo  {journal} {Advances in Physics}\ }\textbf {\bibinfo
  {volume} {63}},\ \bibinfo {pages} {77} (\bibinfo {year} {2014})},\ \Eprint
  {http://arxiv.org/abs/http://arxiv.org/abs/1405.6694}
  {http://arxiv.org/abs/1405.6694} \BibitemShut {NoStop}%
\bibitem [{\citenamefont {Karrasch}\ \emph {et~al.}(2012)\citenamefont
  {Karrasch}, \citenamefont {Bardarson},\ and\ \citenamefont
  {Moore}}]{KarraschDrude}%
  \BibitemOpen
  \bibfield  {author} {\bibinfo {author} {\bibfnamefont {C.}~\bibnamefont
  {Karrasch}}, \bibinfo {author} {\bibfnamefont {J.~H.}\ \bibnamefont
  {Bardarson}}, \ and\ \bibinfo {author} {\bibfnamefont {J.~E.}\ \bibnamefont
  {Moore}},\ }\href@noop {} {\bibfield  {journal} {\bibinfo  {journal} {Phys.
  Rev. Lett.}\ }\textbf {\bibinfo {volume} {108}},\ \bibinfo {pages} {227206}
  (\bibinfo {year} {2012})}\BibitemShut {NoStop}%
\bibitem [{\citenamefont {Karrasch}\ \emph {et~al.}(2013)\citenamefont
  {Karrasch}, \citenamefont {Bardarson},\ and\ \citenamefont
  {Moore}}]{KarraschTransport}%
  \BibitemOpen
  \bibfield  {author} {\bibinfo {author} {\bibfnamefont {C.}~\bibnamefont
  {Karrasch}}, \bibinfo {author} {\bibfnamefont {J.~H.}\ \bibnamefont
  {Bardarson}}, \ and\ \bibinfo {author} {\bibfnamefont {J.~E.}\ \bibnamefont
  {Moore}},\ }\href@noop {} {\bibfield  {journal} {\bibinfo  {journal} {New J.
  Phys.}\ }\textbf {\bibinfo {volume} {15}},\ \bibinfo {pages} {083031}
  (\bibinfo {year} {2013})}\BibitemShut {NoStop}%
\bibitem [{\citenamefont {Pi{\v z}orn}\ \emph {et~al.}(2014)\citenamefont
  {Pi{\v z}orn}, \citenamefont {Eisler}, \citenamefont {Andergassen},\ and\
  \citenamefont {Troyer}}]{PizornfiniteTemp}%
  \BibitemOpen
  \bibfield  {author} {\bibinfo {author} {\bibfnamefont {I.}~\bibnamefont
  {Pi{\v z}orn}}, \bibinfo {author} {\bibfnamefont {V.}~\bibnamefont {Eisler}},
  \bibinfo {author} {\bibfnamefont {S.}~\bibnamefont {Andergassen}}, \ and\
  \bibinfo {author} {\bibfnamefont {M.}~\bibnamefont {Troyer}},\ }\href@noop {}
  {\bibfield  {journal} {\bibinfo  {journal} {New J. Phys.}\ }\textbf {\bibinfo
  {volume} {16}},\ \bibinfo {pages} {073007} (\bibinfo {year}
  {2014})}\BibitemShut {NoStop}%
\bibitem [{\citenamefont {Cui}\ \emph {et~al.}(2015)\citenamefont {Cui},
  \citenamefont {Cirac},\ and\ \citenamefont {Ba\~nuls}}]{JianMariCarmen}%
  \BibitemOpen
  \bibfield  {author} {\bibinfo {author} {\bibfnamefont {J.}~\bibnamefont
  {Cui}}, \bibinfo {author} {\bibfnamefont {J.~I.}\ \bibnamefont {Cirac}}, \
  and\ \bibinfo {author} {\bibfnamefont {M.~C.}\ \bibnamefont {Ba\~nuls}},\
  }\href@noop {} {\bibfield  {journal} {\bibinfo  {journal} {Phys. Rev. Lett.}\
  }\textbf {\bibinfo {volume} {114}},\ \bibinfo {pages} {220601} (\bibinfo
  {year} {2015})}\BibitemShut {NoStop}%
\bibitem [{\citenamefont {Cai}\ and\ \citenamefont
  {Barthel}(2013)}]{CaiBarthel}%
  \BibitemOpen
  \bibfield  {author} {\bibinfo {author} {\bibfnamefont {Z.}~\bibnamefont
  {Cai}}\ and\ \bibinfo {author} {\bibfnamefont {T.}~\bibnamefont {Barthel}},\
  }\href@noop {} {\bibfield  {journal} {\bibinfo  {journal} {Phys. Rev. Lett.}\
  }\textbf {\bibinfo {volume} {111}},\ \bibinfo {pages} {150403} (\bibinfo
  {year} {2013})}\BibitemShut {NoStop}%
\bibitem [{\citenamefont {Kliesch}\ \emph {et~al.}(2014)\citenamefont
  {Kliesch}, \citenamefont {Gross},\ and\ \citenamefont
  {Eisert}}]{UndecidableMPO}%
  \BibitemOpen
  \bibfield  {author} {\bibinfo {author} {\bibfnamefont {M.}~\bibnamefont
  {Kliesch}}, \bibinfo {author} {\bibfnamefont {D.}~\bibnamefont {Gross}}, \
  and\ \bibinfo {author} {\bibfnamefont {J.}~\bibnamefont {Eisert}},\
  }\href@noop {} {\bibfield  {journal} {\bibinfo  {journal} {Phys. Rev. Lett.}\
  }\textbf {\bibinfo {volume} {113}},\ \bibinfo {pages} {160503} (\bibinfo
  {year} {2014})}\BibitemShut {NoStop}%
\bibitem [{\citenamefont {Pichler}\ \emph {et~al.}(2013)\citenamefont
  {Pichler}, \citenamefont {Schachenmayer}, \citenamefont {Daley},\ and\
  \citenamefont {Zoller}}]{DaleyZoller}%
  \BibitemOpen
  \bibfield  {author} {\bibinfo {author} {\bibfnamefont {H.}~\bibnamefont
  {Pichler}}, \bibinfo {author} {\bibfnamefont {J.}~\bibnamefont
  {Schachenmayer}}, \bibinfo {author} {\bibfnamefont {A.~J.}\ \bibnamefont
  {Daley}}, \ and\ \bibinfo {author} {\bibfnamefont {P.}~\bibnamefont
  {Zoller}},\ }\href@noop {} {\bibfield  {journal} {\bibinfo  {journal} {Phys.
  Rev. A}\ }\textbf {\bibinfo {volume} {87}},\ \bibinfo {pages} {033606}
  (\bibinfo {year} {2013})}\BibitemShut {NoStop}%
\bibitem [{\citenamefont {Sarkar}\ \emph {et~al.}(2013)\citenamefont {Sarkar},
  \citenamefont {Langer}, \citenamefont {Schachenmayer},\ and\ \citenamefont
  {Daley}}]{DissipativeLightScattering}%
  \BibitemOpen
  \bibfield  {author} {\bibinfo {author} {\bibfnamefont {S.}~\bibnamefont
  {Sarkar}}, \bibinfo {author} {\bibfnamefont {S.}~\bibnamefont {Langer}},
  \bibinfo {author} {\bibfnamefont {J.}~\bibnamefont {Schachenmayer}}, \ and\
  \bibinfo {author} {\bibfnamefont {A.~J.}\ \bibnamefont {Daley}},\ }\href@noop
  {} {\bibfield  {journal} {\bibinfo  {journal} {Phys. Rev. A}\ }\textbf
  {\bibinfo {volume} {90}},\ \bibinfo {pages} {023618} (\bibinfo {year}
  {2013})}\BibitemShut {NoStop}%
\bibitem [{\citenamefont {Bonnes}\ and\ \citenamefont
  {Laeuchli}(2014)}]{LauchCompare}%
  \BibitemOpen
  \bibfield  {author} {\bibinfo {author} {\bibfnamefont {L.}~\bibnamefont
  {Bonnes}}\ and\ \bibinfo {author} {\bibfnamefont {A.~M.}\ \bibnamefont
  {Laeuchli}},\ }\href@noop {} {} (\bibinfo {year} {2014}),\ \bibinfo {note}
  {arXiv:1411.4831}\BibitemShut {NoStop}%
\bibitem [{\citenamefont {de~las Cuevas}\ \emph {et~al.}(2013)\citenamefont
  {de~las Cuevas}, \citenamefont {Schuch}, \citenamefont {Perez-Garcia},\ and\
  \citenamefont {Cirac}}]{Gemma}%
  \BibitemOpen
  \bibfield  {author} {\bibinfo {author} {\bibfnamefont {G.}~\bibnamefont
  {de~las Cuevas}}, \bibinfo {author} {\bibfnamefont {N.}~\bibnamefont
  {Schuch}}, \bibinfo {author} {\bibfnamefont {D.}~\bibnamefont
  {Perez-Garcia}}, \ and\ \bibinfo {author} {\bibfnamefont {J.~I.}\
  \bibnamefont {Cirac}},\ }\href@noop {} {\bibfield  {journal} {\bibinfo
  {journal} {New J. Phys.}\ }\textbf {\bibinfo {volume} {15}},\ \bibinfo
  {pages} {123021} (\bibinfo {year} {2013})}\BibitemShut {NoStop}%
\bibitem [{\citenamefont {Pi{\v{z}}orn}\ \emph {et~al.}(2011)\citenamefont
  {Pi{\v{z}}orn}, \citenamefont {Wang},\ and\ \citenamefont
  {Verstraete}}]{Pizorn2011time}%
  \BibitemOpen
  \bibfield  {author} {\bibinfo {author} {\bibfnamefont {I.}~\bibnamefont
  {Pi{\v{z}}orn}}, \bibinfo {author} {\bibfnamefont {L.}~\bibnamefont {Wang}},
  \ and\ \bibinfo {author} {\bibfnamefont {F.}~\bibnamefont {Verstraete}},\
  }\href@noop {} {\bibfield  {journal} {\bibinfo  {journal} {Phys. Rev. A}\
  }\textbf {\bibinfo {volume} {83}},\ \bibinfo {pages} {052321} (\bibinfo
  {year} {2011})}\BibitemShut {NoStop}%
\bibitem [{\citenamefont {Lubasch}\ \emph
  {et~al.}(2014{\natexlab{b}})\citenamefont {Lubasch}, \citenamefont {Cirac},\
  and\ \citenamefont {Ba{\~n}uls}}]{Lubasch2014b}%
  \BibitemOpen
  \bibfield  {author} {\bibinfo {author} {\bibfnamefont {M.}~\bibnamefont
  {Lubasch}}, \bibinfo {author} {\bibfnamefont {J.~I.}\ \bibnamefont {Cirac}},
  \ and\ \bibinfo {author} {\bibfnamefont {M.-C.}\ \bibnamefont {Ba{\~n}uls}},\
  }\href@noop {} {\bibfield  {journal} {\bibinfo  {journal} {New J. Phys.}\
  }\textbf {\bibinfo {volume} {16}},\ \bibinfo {pages} {033014} (\bibinfo
  {year} {2014}{\natexlab{b}})}\BibitemShut {NoStop}%
\bibitem [{\citenamefont {Pirvu}\ \emph {et~al.}(2010)\citenamefont {Pirvu},
  \citenamefont {Murg}, \citenamefont {Cirac},\ and\ \citenamefont
  {Verstraete}}]{MPO_Representations}%
  \BibitemOpen
  \bibfield  {author} {\bibinfo {author} {\bibfnamefont {B.}~\bibnamefont
  {Pirvu}}, \bibinfo {author} {\bibfnamefont {V.}~\bibnamefont {Murg}},
  \bibinfo {author} {\bibfnamefont {J.~I.}\ \bibnamefont {Cirac}}, \ and\
  \bibinfo {author} {\bibfnamefont {F.}~\bibnamefont {Verstraete}},\
  }\href@noop {} {\bibfield  {journal} {\bibinfo  {journal} {New J. Phys.}\
  }\textbf {\bibinfo {volume} {12}},\ \bibinfo {pages} {025012} (\bibinfo
  {year} {2010})}\BibitemShut {NoStop}%
\bibitem [{\citenamefont {Vidal}(2003)}]{GVidal}%
  \BibitemOpen
  \bibfield  {author} {\bibinfo {author} {\bibfnamefont {G.}~\bibnamefont
  {Vidal}},\ }\href@noop {} {\bibfield  {journal} {\bibinfo  {journal} {Phys.
  Rev. Lett.}\ }\textbf {\bibinfo {volume} {91}},\ \bibinfo {pages} {147902}
  (\bibinfo {year} {2003})}\BibitemShut {NoStop}%
\bibitem [{\citenamefont {Kliesch}\ \emph {et~al.}(2011)\citenamefont
  {Kliesch}, \citenamefont {Barthel}, \citenamefont {Gogolin}, \citenamefont
  {Kastoryano},\ and\ \citenamefont {Eisert}}]{KliBarGog11}%
  \BibitemOpen
  \bibfield  {author} {\bibinfo {author} {\bibfnamefont {M.}~\bibnamefont
  {Kliesch}}, \bibinfo {author} {\bibfnamefont {T.}~\bibnamefont {Barthel}},
  \bibinfo {author} {\bibfnamefont {C.}~\bibnamefont {Gogolin}}, \bibinfo
  {author} {\bibfnamefont {M.}~\bibnamefont {Kastoryano}}, \ and\ \bibinfo
  {author} {\bibfnamefont {J.}~\bibnamefont {Eisert}},\ }\href {\doibase
  10.1103/PhysRevLett.107.120501} {\bibfield  {journal} {\bibinfo  {journal}
  {Phys. Rev. Lett.}\ }\textbf {\bibinfo {volume} {107}},\ \bibinfo {pages}
  {120501} (\bibinfo {year} {2011})}\BibitemShut {NoStop}%
\bibitem [{\citenamefont {Schmidt}\ \emph {et~al.}(2010)\citenamefont
  {Schmidt}, \citenamefont {Gerace}, \citenamefont {Houck}, \citenamefont
  {Blatter},\ and\ \citenamefont {T{\"u}reci}}]{Schmidt2010}%
  \BibitemOpen
  \bibfield  {author} {\bibinfo {author} {\bibfnamefont {S.}~\bibnamefont
  {Schmidt}}, \bibinfo {author} {\bibfnamefont {D.}~\bibnamefont {Gerace}},
  \bibinfo {author} {\bibfnamefont {A.~A.}\ \bibnamefont {Houck}}, \bibinfo
  {author} {\bibfnamefont {G.}~\bibnamefont {Blatter}}, \ and\ \bibinfo
  {author} {\bibfnamefont {H.~E.}\ \bibnamefont {T{\"u}reci}},\ }\href@noop {}
  {\bibfield  {journal} {\bibinfo  {journal} {Phys. Rev. B}\ }\textbf {\bibinfo
  {volume} {82}},\ \bibinfo {pages} {100507} (\bibinfo {year}
  {2010})}\BibitemShut {NoStop}%
\bibitem [{\citenamefont {Poletti}\ \emph {et~al.}(2013)\citenamefont
  {Poletti}, \citenamefont {Barmettler}, \citenamefont {Georges},\ and\
  \citenamefont {Kollath}}]{GlassyKollath}%
  \BibitemOpen
  \bibfield  {author} {\bibinfo {author} {\bibfnamefont {D.}~\bibnamefont
  {Poletti}}, \bibinfo {author} {\bibfnamefont {P.}~\bibnamefont {Barmettler}},
  \bibinfo {author} {\bibfnamefont {A.}~\bibnamefont {Georges}}, \ and\
  \bibinfo {author} {\bibfnamefont {C.}~\bibnamefont {Kollath}},\ }\href@noop
  {} {\bibfield  {journal} {\bibinfo  {journal} {Phys. Rev. Lett.}\ }\textbf
  {\bibinfo {volume} {111}},\ \bibinfo {pages} {195301} (\bibinfo {year}
  {2013})}\BibitemShut {NoStop}%
\bibitem [{\citenamefont {Lin}\ \emph {et~al.}(2013)\citenamefont {Lin},
  \citenamefont {Gaebler}, \citenamefont {Reiter}, \citenamefont {Tan},
  \citenamefont {Bowler}, \citenamefont {S{\/o}rensen}, \citenamefont
  {Leibfried},\ and\ \citenamefont {Wineland}}]{Dissipengineer}%
  \BibitemOpen
  \bibfield  {author} {\bibinfo {author} {\bibfnamefont {Y.}~\bibnamefont
  {Lin}}, \bibinfo {author} {\bibfnamefont {J.~P.}\ \bibnamefont {Gaebler}},
  \bibinfo {author} {\bibfnamefont {F.}~\bibnamefont {Reiter}}, \bibinfo
  {author} {\bibfnamefont {T.~R.}\ \bibnamefont {Tan}}, \bibinfo {author}
  {\bibfnamefont {R.}~\bibnamefont {Bowler}}, \bibinfo {author} {\bibfnamefont
  {A.~S.}\ \bibnamefont {S{\/o}rensen}}, \bibinfo {author} {\bibfnamefont
  {D.}~\bibnamefont {Leibfried}}, \ and\ \bibinfo {author} {\bibfnamefont
  {D.~J.}\ \bibnamefont {Wineland}},\ }\href {\doibase 10.1038/nature12801}
  {\bibfield  {journal} {\bibinfo  {journal} {Nature}\ }\textbf {\bibinfo
  {volume} {504}},\ \bibinfo {pages} {415} (\bibinfo {year}
  {2013})}\BibitemShut {NoStop}%
\bibitem [{\citenamefont {Doria}\ \emph {et~al.}(2011)\citenamefont {Doria},
  \citenamefont {Calarco},\ and\ \citenamefont {Montangero}}]{DoriaControl}%
  \BibitemOpen
  \bibfield  {author} {\bibinfo {author} {\bibfnamefont {P.}~\bibnamefont
  {Doria}}, \bibinfo {author} {\bibfnamefont {T.}~\bibnamefont {Calarco}}, \
  and\ \bibinfo {author} {\bibfnamefont {S.}~\bibnamefont {Montangero}},\
  }\href@noop {} {\bibfield  {journal} {\bibinfo  {journal} {Phys. Rev. Lett.}\
  }\textbf {\bibinfo {volume} {106}},\ \bibinfo {pages} {190501} (\bibinfo
  {year} {2011})}\BibitemShut {NoStop}%
\bibitem [{\citenamefont {Vacanti}\ \emph {et~al.}(2014)\citenamefont
  {Vacanti}, \citenamefont {Fazio}, \citenamefont {Montangero}, \citenamefont
  {Palma}, \citenamefont {Paternostro},\ and\ \citenamefont
  {Vedral}}]{Shortcuts}%
  \BibitemOpen
  \bibfield  {author} {\bibinfo {author} {\bibfnamefont {G.}~\bibnamefont
  {Vacanti}}, \bibinfo {author} {\bibfnamefont {R.}~\bibnamefont {Fazio}},
  \bibinfo {author} {\bibfnamefont {S.}~\bibnamefont {Montangero}}, \bibinfo
  {author} {\bibfnamefont {G.~M.}\ \bibnamefont {Palma}}, \bibinfo {author}
  {\bibfnamefont {M.}~\bibnamefont {Paternostro}}, \ and\ \bibinfo {author}
  {\bibfnamefont {V.}~\bibnamefont {Vedral}},\ }\href@noop {} {\bibfield
  {journal} {\bibinfo  {journal} {New J. Phys.}\ }\textbf {\bibinfo {volume}
  {16}},\ \bibinfo {pages} {053017} (\bibinfo {year} {2014})}\BibitemShut
  {NoStop}%
\bibitem [{\citenamefont {Kastoryano}\ and\ \citenamefont
  {Eisert}(2013)}]{Mixing}%
  \BibitemOpen
  \bibfield  {author} {\bibinfo {author} {\bibfnamefont {M.~J.}\ \bibnamefont
  {Kastoryano}}\ and\ \bibinfo {author} {\bibfnamefont {J.}~\bibnamefont
  {Eisert}},\ }\href@noop {} {\bibfield  {journal} {\bibinfo  {journal} {J.
  Math. Phys.}\ }\textbf {\bibinfo {volume} {54}},\ \bibinfo {pages} {102201}
  (\bibinfo {year} {2013})}\BibitemShut {NoStop}%
\bibitem [{\citenamefont {Cubitt}\ \emph {et~al.}(2015)\citenamefont {Cubitt},
  \citenamefont {Lucia}, \citenamefont {Michalakis},\ and\ \citenamefont
  {Perez-Garcia}}]{MadridRobustness}%
  \BibitemOpen
  \bibfield  {author} {\bibinfo {author} {\bibfnamefont {T.~S.}\ \bibnamefont
  {Cubitt}}, \bibinfo {author} {\bibfnamefont {A.}~\bibnamefont {Lucia}},
  \bibinfo {author} {\bibfnamefont {S.}~\bibnamefont {Michalakis}}, \ and\
  \bibinfo {author} {\bibfnamefont {D.}~\bibnamefont {Perez-Garcia}},\
  }\href@noop {} {\bibfield  {journal} {\bibinfo  {journal} {Commun. Math.
  Phys.}\ }\textbf {\bibinfo {volume} {337}},\ \bibinfo {pages} {1275}
  (\bibinfo {year} {2015})}\BibitemShut {NoStop}%
\bibitem [{\citenamefont {{Watrous}}(2012)}]{Wat12}%
  \BibitemOpen
  \bibfield  {author} {\bibinfo {author} {\bibfnamefont {J.}~\bibnamefont
  {{Watrous}}},\ }\href@noop {} {\enquote {\bibinfo {title} {Simpler
  semidefinite programs for completely bounded norms},}\ }\ \Eprint {http://arxiv.org/abs/1207.5726} {arXiv:1207.5726} \BibitemShut {NoStop}%
\bibitem [{\citenamefont {Suzuki}(1991)}]{Suz91}%
  \BibitemOpen
  \bibfield  {author} {\bibinfo {author} {\bibfnamefont {M.}~\bibnamefont
  {Suzuki}},\ }\href {\doibase 10.1063/1.529425} {\bibfield  {journal}
  {\bibinfo  {journal} {J. Math. Phys.}\ }\textbf {\bibinfo {volume} {32}},\
  \bibinfo {pages} {400} (\bibinfo {year} {1991})}\BibitemShut {NoStop}%
\bibitem [{\citenamefont {Schollwoeck}(2011)}]{Schollwock201196}%
  \BibitemOpen
  \bibfield  {author} {\bibinfo {author} {\bibfnamefont {U.}~\bibnamefont
  {Schollwoeck}},\ }\href@noop {} {\bibfield  {journal} {\bibinfo  {journal}
  {Ann. Phys.}\ }\textbf {\bibinfo {volume} {326}},\ \bibinfo {pages} {96}
  (\bibinfo {year} {2011})}\BibitemShut {NoStop}%
\bibitem [{\citenamefont {Choi}(1975)}]{Choimap}%
  \BibitemOpen
  \bibfield  {author} {\bibinfo {author} {\bibfnamefont {M.-D.}\ \bibnamefont
  {Choi}},\ }\href@noop {} {\bibfield  {journal} {\bibinfo  {journal} {Lin.
  Alg. App.}\ }\textbf {\bibinfo {volume} {10}},\ \bibinfo {pages} {285}
  (\bibinfo {year} {1975})}\BibitemShut {NoStop}%
\bibitem [{\citenamefont {Schuch}\ \emph {et~al.}(2007)\citenamefont {Schuch},
  \citenamefont {Wolf}, \citenamefont {Verstraete},\ and\ \citenamefont
  {Cirac}}]{SchWolVer07}%
  \BibitemOpen
  \bibfield  {author} {\bibinfo {author} {\bibfnamefont {N.}~\bibnamefont
  {Schuch}}, \bibinfo {author} {\bibfnamefont {M.~M.}\ \bibnamefont {Wolf}},
  \bibinfo {author} {\bibfnamefont {F.}~\bibnamefont {Verstraete}}, \ and\
  \bibinfo {author} {\bibfnamefont {J.~I.}\ \bibnamefont {Cirac}},\ }\href
  {\doibase 10.1103/PhysRevLett.98.140506} {\bibfield  {journal} {\bibinfo
  {journal} {Phys. Rev. Lett.}\ }\textbf {\bibinfo {volume} {98}},\ \bibinfo
  {pages} {140506} (\bibinfo {year} {2007})}\BibitemShut {NoStop}%
\end{thebibliography}
\end{document}